\documentclass{amsart}[11pt]
\usepackage{amsmath, amsfonts, amsthm, amssymb,mathtools}
\usepackage{subfigure}
\usepackage{caption} 
\usepackage{stmaryrd}
\usepackage{verbatim}
\usepackage{hyperref}
\usepackage{color}
\usepackage{ulem}
\usepackage{enumerate}
\numberwithin{equation}{section}
\newtheorem{theorem}{Theorem}[section]

\newtheorem{lemma}[theorem]{Lemma}
\newtheorem{proposition}[theorem]{Proposition}

\theoremstyle{definition}
\newtheorem{definition}[theorem]{Definition}
\newtheorem{remark}[theorem]{Remark}
\newtheorem{algorithm}[theorem]{Algorithm}
\newtheorem{assumption}[theorem]{Assumption}

\newcommand{\QQ}{\mathbb{Q}}
\newcommand{\RR}{\mathbb{R}}

\newcommand{\PP}{\mathbb{P}}
\newcommand{\VV}{\mathbb{V}}
\newcommand{\NN}{\mathbb{N}}
\newcommand{\Fk}{\mathfrak{F}}
\newcommand{\Vk}{\mathfrak{V}}

\newcommand{\D}{\mathrm{d}}
\newcommand{\BSS}{\mathcal{BSS}}
\newcommand{\Bb}{\mathcal{B}}
\newcommand{\Cc}{\mathcal{C}}

\newcommand{\Ee}{\mathcal{E}}
\newcommand{\Ll}{\mathcal{L}}
\newcommand{\Tt}{\mathcal{T}}
\newcommand{\Cr}{\mathrm{C}}
\newcommand{\F}{\mathrm{F}}
\newcommand{\E}{\mathrm{e}}
\newcommand{\BS}{\mathrm{BS}}
\newcommand{\VIX}{\mathrm{VIX}}

\newcommand{\EE}{\mathbb{E}}

\newcommand{\Nn}{\mathcal{N}}
\newcommand{\Vv}{\mathcal{V}}
\newcommand{\Ff}{\mathcal{F}}
\newcommand{\Oo}{\mathcal{O}}
\newcommand{\eps}{\varepsilon}
\newcommand{\tmu}{\tilde{\mu}}
\newcommand{\tsigma}{\tilde{\sigma}}
\newcommand{\Hm}{H_{-}}
\newcommand{\Hp}{H_{+}}

\def\equalDistrib{\,{\buildrel \Delta \over =}\,}

\begin{document}
\title{On VIX Futures in the rough Bergomi model}
\date{\today}
\author{Antoine Jacquier}
\address{Department of Mathematics, Imperial College London}
\email{a.jacquier@imperial.ac.uk}
\author{Claude Martini}
\address{Zeliade Systems, Paris}
\email{c.martini@zeliade.com}
\author{Aitor Muguruza}
\address{Department of Mathematics, Imperial College London}
\email{aitor.muguruza-gonzalez15@imperial.ac.uk}
\thanks{The authors would like to thank Christian Bayer, Jim Gatheral, Mikko Pakkanen and Mathieu Rosenbaum for useful discussions.
AJ acknowledges financial support from the EPSRC First Grant EP/M008436/1.
The numerical implementations have been carried out on the collaborative platform Zanadu (www.zanadu.io).
}
\maketitle
\begin{abstract}
The rough Bergomi model introduced by Bayer, Friz and Gatheral~\cite{BFG15} 
has been outperforming conventional Markovian stochastic volatility models 
by reproducing implied volatility smiles in a very realistic manner, in particular for short maturities. 
We investigate here the dynamics of the VIX and the forward variance curve generated by this model,
and develop efficient pricing algorithms for VIX futures and options.
We further analyse the validity of the rough Bergomi model to jointly describe the VIX and the SPX, 
and present a joint calibration algorithm based on the hybrid scheme by Bennedsen, Lunde and Pakkanen~\cite{BLP15}. 
\end{abstract}



\section{Introduction}
Volatility, though not directly observed nor traded, is a fundamental object on financial markets,
and has been the centre of attention of decades of theoretical and practical research,
both to estimate it and to use it for trading purposes.
The former goal has usually been carried out under the historical measure ($\PP$)
while the latter, through the introduction of volatility derivatives (VIX and related family),
has been evolving under the pricing measure~$\QQ$.
Most models used for pricing purposes (Heston~\cite{Heston}, SABR~\cite{SABR}, Bergomi~\cite{Bergomi}) are constructed under~$\QQ$ and are of Markovian nature (making pricing, and hence calibration, easier). 
Recently, Gatheral, Jaisson and Rosenbaum~\cite{Volrough} broke this routine
and introduced a fractional Brownian motion as driving factor of the volatility process. 
This approach (Rough Fractional Stochastic Volatility, RFSV for short) 
opens the door to revisiting classical pricing and calibration conundrums.
They, together with the subsequent paper by Bayer, Friz and Gatheral, (see also~\cite{Alos, Fukasawa}) 
in particular showed that these models were able to capture the extra steepness of the implied volatility smile
in Equity markets for short maturities, which continuous Markovian stochastic volatility models fail to describe.
The icing on the cake is the (at last!!) reconciliation between the two measures~$\PP$ and~$\QQ$ 
within a given model, showing remarkable results both for estimation and for prediction. 

One of the key issues in Equity markets is, not only to fit the (SPX) implied volatility smile, 
but to do so jointly with a calibration of the VIX (Futures and ideally options).
Gatheral's~\cite{G2008} double mean reverting process is the leading (Markovian) continuous model in this direction,
while models with jumps have been proposed abundantly by Carr and Madan~\cite{CM14} 
and Kokholm and Stisen~\cite{KS15}. 
This issue was briefly tackled by Bayer, Friz and Gatheral~\cite{BFG15} for a particular rough model
(rough Bergomi),
and we aim here at providing a deeper analysis of VIX dynamics under this rough model
and at implementing pricing schemes for VIX Futures and options. 
Our main contribution is a precise link between the forward variance curve~$(\xi_T(\cdot))_{T\geq 0}$ 
and the initial forward variance curve~$\xi_0(\cdot)$ in the rough Bergomi model. 
This in turn, allows us not only to provide simulation methods for the VIX,
but also to refine the log-normal approximation of~\cite{BFG15} for VIX Futures, 
matching exactly the first two moments. 
Finally, we develop an efficient algorithm for VIX Futures calibration, 
upon which we build a joint calibration method with the SPX. 
As opposed to the Cholesky approach in~\cite{BFG15}, 
we adapt the hybrid-scheme by Bennedsen, Lunde and Pakkanen~\cite{BLP15} 
with better complexity $\mathcal{O} (n\log n)$ . 
Assuming the universality of the Hurst parameter~$H$ across VIX and SPX 
allows us to compute efficiently prices recursively with complexity $\mathcal{O} (n)$. 
In passing, we also investigate the joint consistency of VIX and SPX in the market.
The organisation of the paper follows accordingly:
we first introduce the rough Bergomi model and its main properties (Section~\ref{sec:rBergomi}), 
before presenting its pricing power for VIX Futures (Section~\ref{sec:VIXFutures}),
and finally develop the joint calibration algorithm in Section~\ref{sec:JointCalib}.


\section{Rough volatility and the rough Bergomi model}\label{sec:rBergomi}
Comte and Renault\cite{Comte-Renault} were the first to propose a stochastic volatility model 
in which the instantaneous volatility is driven by a fractional Brownian motion, 
with a Hurst index restricted to be greater than~$1/2$. 
Recently Gatheral, Jaisson and Rosenbaum~\cite{Volrough} presented a new approach 
with a Hurst index smaller than~$1/2$, 
producing extremely good fits to observed volatility data under the physical measure~$\PP$. 
These models form the so-called Rough Fractional Stochastic Volatility (RFSV) family that is understood as a natural extension of the classical volatility models driven by standard Brownian motion.
Our work focuses on the pricing measure $\QQ$ and we assume through this paper that the model presented by Gatheral, Jaisson and Rosenbaum~\cite{Volrough} under $\PP$ is a reasonable model. 
Finally, and most importantly, we follow the recent paper by Bayer, Friz and Gatheral~\cite{BFG15}, in order to extend the RFSV model to pricing schemes under the measure $\QQ$. 
More precisely, Bayer, Friz and Gatheral~\cite{BFG15} proposed the following model for the log stock price process $X:=\log(S)$:
\begin{equation}\label{eq:rBergomi}
\begin{array}{rll}
\D X_t & = \displaystyle -\frac{1}{2}V_t \D t + \sqrt{V_t}\D W_t, & X_0 = 0\\
V_t & =  \displaystyle \xi_0(t)\Ee(2\nu C_{H}\Vv_t), & V_0>0,
\end{array}
\end{equation}
with $\nu,\xi_0(\cdot)>0$, $\Ee(\cdot)$ is the Dol\'eans-Dade~\cite{Doleans} 
stochastic exponential
and $C_H:=\sqrt{\frac{2H\Gamma(2-\Hp)}{\Gamma(\Hp)\Gamma(2-2H)}}$,
where, for notational convenience (throughout the paper), we use the symbols $H_{\pm} := H\pm \frac{1}{2}$.
All the processes are defined on a given filtered probability space $(\Omega, \Ff, (\Ff_t)_{t\geq 0}, \QQ)$
supporting the two standard Brownian motions~$W$ and~$Z$ (see below).
The initial forward variance curve is observed at inception, 
and we therefore assume without loss of generality that it is $\Ff_0$-measurable.
The process~$\Vv$, defined as 
\begin{equation}\label{eq:Volterra}
\Vv_{t} := \int_{0}^{t} (t-u)^{\Hm}\D Z_u,
\end{equation}
is a centred Gaussian process with covariance structure
$$
\EE(\Vv_t \Vv_s) = s^{2H}\int_{0}^{1} \left(\frac{t}{s}-u\right)^{\Hm}(1-u)^{\Hm}\D u,
\qquad\text{for any }s,t \in [0,1].
$$
We shall also introduce, for any $0 \leq T\leq t$ the notations
\begin{equation}\label{eq:VtT}
\Vv_{t, T} := \int_{T}^{t} (t-u)^{\Hm}\D Z_u
\qquad\text{and}\qquad
\Vv_t^T:=\int_{0}^{T}(t-u)^{\Hm}\D Z_u.
\end{equation}
Note in particular that $\Vv^T_T = \Vv_T$.
The two standard Brownian motions~$W$ and~$Z$ are correlated with correlation parameter $\rho \in (-1,1)$.
Here, $\xi_T(t)$ denotes the forward variance observed at time~$T$ for a maturity equal to~$t$.
More precisely, if $\sigma^2_T(t)$ denotes the fair strike of a variance swap observed at time~$T$ 
and maturing at~$t$, then 
$$
\sigma^2_T(t)=\frac{1}{t-T}\int_T^t\xi_T(u)\D u,
$$
or equivalently
$\xi_T(t)=\frac{\D}{\D t}\left((t-T)\sigma^2_T(t)\right)$.
For any fixed $t>0$, the process $(\xi_s(t))_{s\leq t}$, is a martingale, i.e.
$\EE[\xi_s(t)|\Ff_u]=\xi_u(t)$,
for all $u\leq s\leq t$.
Furthermore,
$\Vv^T_t$ is a centred Gaussian process with variance
\begin{equation}\label{eq:VtTVariance}
\VV(\Vv^T_t) = \frac{t^{2H}-(t-T)^{2H}}{2H}, 
\qquad\text{for }t\geq T,
\end{equation}
and covariance structure
\begin{equation}\label{eq:CovStructure}
\EE\left(\Vv^T_t \Vv^T_s\right)
 = \int_0^T\left[(t-u)(s-u)\right]^{\Hm}\D u
 = \frac{(s - t)^{\Hm}}{\Hp}\left\{
 t^{\Hp}\F\left(\frac{-t}{s - t}\right) - (t - T)^{\Hp}\F\left(\frac{T-t}{s - t}\right) \right\},
\end{equation}
for any $t<s$, where we introduce the function $\F:\RR_-\to\RR$ as
\begin{equation}\label{eq:F}
\F(u) := \text{}_2F_1\left(-\Hm, \Hp , 1+\Hp , u\right),
\end{equation}
and $\text{}_2F_1$ is the hypergeometric function~\cite[Chapter 15]{Abramowitz-Stegun}. 
Finally, the quadratic variation of $\Vv$ is given by
\begin{equation}\label{eq: quadratic_variation}
[\Vv]_t= \frac{t^{2H}}{2H}, 
\qquad\text{for }t\geq 0.
\end{equation}

\subsection{Hybrid simulation scheme}\label{Hybrid simulation scheme}
Bayer, Friz and Gatheral~\cite{BFG15} present a Cholesky method to simulate the rough Bergomi model. 
Although exact, this method is very slow and other approaches need to be considered for calibration purposes.
Recently, Bennedsen, Lunde and Pakkanen~\cite{BLP15} presented a new simulation scheme for Brownian semistationary ($\BSS$) processes. 
This method, as opposed to Cholesky, is an approximate method. 
However, in~\cite{BLP15} the authors show that the method yields remarkable results 
in the case of the rough Bergomi model. 
In addition, their approach leads to a natural simulation of both the Volterra process~$\Vv$ 
and the stock price~$S$ and yields a computational complexity of order~$\Oo(n\log n)$.

\begin{definition} \label{def:semistationary}
Let $W$ be a standard Brownian motion on a given filtered probability space $(\Omega, \Ff, (\Ff_t)_{t\geq 0},\mathbb{P})$. A truncated Brownian semistationary ($\BSS$) process is defined as
$\Bb(t) = \int_{0}^{t} g(t-s)\sigma(s)\D W_s$, for $t\geq 0$,
where~$\sigma$ is $(\Ff_t)_{t\geq 0}$-predictable with locally bounded trajectories and finite second moments, and $g:(0,\infty)\to[0,\infty)$ is Borel measurable and square integrable.
We shall call it a $\BSS(\alpha, W)$ process if furthermore
\begin{enumerate}[(i)]
\item\label{A1} there exists $\alpha \in\left(-\frac{1}{2},\frac{1}{2}\right)\setminus\{0\}$ such that
$g(x)=x^\alpha L_g(x)$ for all $x\in(0,1]$,
where $L_g\in \Cc^1((0,1]\to [0,\infty))$, is slowly varying\footnote{A measurable function $L:(0,1]\to[0,\infty)$ is slowly varying~\cite{Bingham} at $0$ if for any $t>0$,
$\lim\limits_{x\downarrow 0}L(tx) / L(x)=1.$}
 at the origin and bounded away from zero.
Moreover, there exists a constant $C>0$ such that 
$|L'_g(x)|\leq C(1+x^{-1})$ for all $x\in(0,1]$;
\item \label{A2} the function~$g$ is differentiable on $(0,\infty)$.
\end{enumerate}
\end{definition}
Under this assumption, the hybrid scheme, proposed in~\cite{BLP15} 
and recalled in Appendix~\ref{sec:HybridScheme},
provides an efficient way to simulate $\BSS$ processes.
It applies in particular to the rough Bergomi model:
\begin{proposition}
The Volterra process~$\Vv$ in~\eqref{eq:Volterra} is a truncated $\BSS(\Hm, Z)$ process.
\end{proposition}
\begin{proof}
From~\eqref{eq:Volterra}, $g(x) \equiv x^{\Hm}$ and $\sigma(\cdot)\equiv 1$ as in Definition~\ref{def:semistationary}, 
so that~$V$ is a $\BSS$ process.
Since $\Hm\in(-\frac{1}{2},0)$, then $L_g \equiv 1$, 
and~$\Vv$ satisfies Definition~\ref{def:semistationary}(i);
Definition~\ref{def:semistationary}(i) trivially holds, and so does the corollary.
\end{proof}

The corollary implies that we can apply the hybrid scheme to~$\Vv$.
In particular,  for $\kappa=1$ the matrix form representation of the scheme reads
(recall that $n_T:=\lfloor nT\rfloor$)
$$
\begin{pmatrix} 
\Vv\left(\frac{1}{n}\right)\\
\Vv\left(\frac{2}{n}\right)\\
\vdots\\
\Vv\left(\frac{n_T}{n}\right)
\end{pmatrix}=
\begin{pmatrix} 
\overline{Z}_{0,1} & 0 & \cdots &  \cdots & 0\\
\overline{Z}_{1,1} & \overline{Z}_0  & \ddots & \ddots & 0\\
\overline{Z}_{2,1}& \overline{Z}_1  & \ddots & \ddots & \vdots\\
\vdots & \ddots & \ddots & \ddots & \vdots\\
\overline{Z}_{n_T-1,1}& \overline{Z}_{n_T-2}  & \cdots & \overline{Z}_1 & \overline{Z}_0
\end{pmatrix} 
\begin{pmatrix}
1\\
\left(\frac{1}{n}b_1^*\right)^{\Hm}\\
\left(\frac{1}{n}b_2^*\right)^{\Hm}\\
\vdots\\
\left(\frac{1}{n}b_{n_T-1}^*\right)^{\Hm}
\end{pmatrix},
$$
where the coefficients $\{b^*_i\}$ are defined in~\eqref{eq:bStar}.
This matrix multiplication is, by brute force, of order $\Oo(n^2)$, 
however using discrete convolution we may use FFT to reduce it to $\Oo(n\log n)$ as suggested in~\cite{BLP15} .

\section{Rough Bergomi and VIX}\label{sec:VIXFutures}
Bayer, Friz and Gatheral~\cite{BFG15} briefly discuss the lack of consistency of the rough Bergomi model 
with observed VIX options data, leading to an incorrect term structure of the VIX. 
In this section, we investigate in detail the dynamics of the VIX,
and propose a log-normal approximation. 
Additionally, we investigate the viability of the model in terms of VIX Futures and options, 
and compare it to the approximation in~\cite{BFG15}.

\subsection{VIX Futures in the rough Bergomi model.}
From now on, we fix a given maturity $T\geq 0$, 
and define the VIX at time~$T$ via the continuous-time monitoring formula
$$
\VIX^2_T := \EE\left(\left.\frac{1}{\Delta}\int_{T}^{T+\Delta}\D\langle X_s, X_s\rangle\D s \right|\Ff_T\right),
$$
where $\Delta$ is equal to $30$ days.
The risk-neutral formula for the VIX future~$\Vk_T$ with maturity~$T$
is then given by
\begin{equation}\label{eq:VIXPrice}
\mathfrak{V}_T:=\EE\left(\VIX_T|\Ff_0\right)
=\EE\left(\left.\sqrt{\frac{1}{\Delta}\int_T^{T+\Delta}\EE\left(\D\langle X_s, X_s\rangle|\Ff_T\right)\D s}\right|\Ff_0\right)
 = \EE\left(\left.\sqrt{\frac{1}{\Delta}\int_T^{T+\Delta}\xi_T(s) \D s}\right|\Ff_0\right).
\end{equation}
Note that, when $T>0$, $\xi_T(s)$ is a market input which is not~$\Ff_0$-measurable, 
and is hence difficult to interpret only knowing~$\Ff_0$.
We shall make repeated use of the following random variable defined for any $t\geq T$, by
\begin{equation}\label{eq:eta}
\eta_T(t) := \exp\left(2\nu C_{H}\Vv_t^T\right).
\end{equation}
\begin{proposition}\label{prop:rBergomiVIX}
The VIX  dynamics are given by
$$
\VIX_T
  = \left\{\frac{1}{\Delta}\int_T^{T+\Delta}\xi_0(t)\eta_T(t)
 \exp\left(\frac{\nu^2C_H^2}{H}\left[(t-T)^{2H}-t^{2H}\right] \right)\D t\right\}^{1/2}.
$$
\end{proposition}
\begin{proof}
Using Fubini's theorem and the instantaneous variance representation in~\eqref{eq:rBergomi}, we can write
\begin{align*}
\VIX^2_T
 &  = \frac{1}{\Delta}\int_T^{T+\Delta}\EE\left(V_s|\Ff_T\right)\D s
 = \frac{1}{\Delta}\int_T^{T+\Delta}\EE\Big[\xi_0(t)\Ee\left(2\nu C_{H}\Vv_t\right)|\Ff_T\Big]\D t\\
 & = \frac{1}{\Delta}\int_T^{T+\Delta}\EE\left[\xi_0(t)\eta_T(t)
\left.\exp\left(2\nu C_{H}\Vv_{t,T}-\frac{\nu^2C_H^2t^{2H}}{H}\right)\right|\Ff_T\right]\D t,
\end{align*}
with $\Vv_{t, T}$ defined in~\eqref{eq:VtT}.
Since $\eta_T(t)\in\Ff_T$ and $\xi_0(t)\in\Ff_0$, this expression simplifies to
$$
\VIX^2_T=\frac{1}{\Delta}\int_T^{T+\Delta}\xi_0(t)\eta_T(t)
\EE\left[\left.\exp\left(2\nu C_{H}\Vv_{t, T}-\frac{\nu^2C_H^2t^{2H}}{H}\right)\right|\Ff_T\right]\D t.
$$
The proposition follows since $\Vv_{t,T}$ is centred Gaussian, independent of~$\Ff_T$,
with variance given in~\eqref{eq:VtTVariance}, and
$
\EE\left(\left.\E^{2\nu C_{H}\Vv_{t,T}}\right|\Ff_T\right) = 
\EE\left(\E^{2\nu C_{H}\Vv_{t,T}}\right) = 
\exp\left(\frac{\nu^2 C_{H}^2}{H}(t-T)^{2H}\right).
$
\end{proof}
The main challenge for simulation is~$\eta_T(t)$. 
However, since the latter is independent of~$\xi_0(\cdot)$, 
robustness of simulation schemes for the VIX will not be affected 
by the qualitative properties of the initial variance curve~$\xi_0$.

\begin{proposition}\label{prop:xiMartingale}
The forward variance curve~$\xi_T$ in the rough Bergomi model admits the representation
$$
\xi_T(t)=\xi_0(t)\eta_T(t)\exp\left(\frac{\nu^2C_H^2}{H}\left[(t-T)^{2H}-t^{2H}\right]\right),
\qquad\text{for any }t\geq T.
$$
\end{proposition}
\begin{proof}
Since $\EE\left(V_t|\Ff_T\right) = \xi_T(t)$ by~\eqref{eq:VIXPrice}, 
the proposition follows from Proposition~\ref{prop:rBergomiVIX} and the equality
$$
\EE\left(V_t|\Ff_T\right)
 = \xi_0(t)\eta_T(t)\exp\left(\frac{\nu^2C_H^2}{H}\left[(t-T)^{2H}-t^{2H}\right]\right).
$$
\end{proof}
Bayer, Friz and Gatheral~\cite{BFG15} did not derive such a representation for~$\xi_T$, 
and their approach for pricing VIX derivatives relies on an approximation 
which avoids the computations developed in this section.
Proposition~\ref{prop:xiMartingale} allows for a better understanding of the process~$\xi_T$, 
and for an innovative approach to price VIX derivatives.

\subsection{Upper and lower bounds for VIX Futures}
\begin{theorem}\label{thm:Bounds}
The following bounds hold for VIX Futures:
\begin{equation}
\frac{1}{\Delta}\int_T^{T+\Delta}\sqrt{\xi_0(t)}
\exp\left(\frac{\nu^2C_H^2}{4H}\left[(t-T)^{2H}-t^{2H}\right]\right)\D t
\leq \Vk_T
\leq \left\{\frac{1}{\Delta}\int_T^{T+\Delta}\xi_0(s)\D s\right\}^{1/2}.
\end{equation}
\end{theorem}
\begin{proof}
The conditional Jensen's inequality gives
$$
\Vk_T
 = \EE\left(\VIX_T|\Ff_0\right)
 = \EE\left(\left.\sqrt{\frac{1}{\Delta}\int_T^{T+\Delta}\xi_T(s)\D s}\right|\Ff_0\right)
 \leq \sqrt{\EE\left(\left.\frac{1}{\Delta}\int_T^{T+\Delta}\xi_T(s)\D s\right|\Ff_0\right)}.
$$
Furthermore, since~$\xi_0$ is $\Ff_0$-adapted, Fubini's theorem along with the martingale property 
of~$\xi_T$ yield the upper bound
$\Vk_T = \EE\left(\VIX_T|\Ff_0\right) \leq \sqrt{\Delta^{-1}\int_T^{T+\Delta}\xi_0(s)\D s}$.
To obtain a lower bound we use the  representation in Proposition~\ref{prop:xiMartingale}, 
and Cauchy-Schwarz's inequality, and Fubini's theorem, so that
\begin{align*}
\Vk_T & = \EE\left(\VIX_T|\Ff_0\right)
 =  \EE\left[\left.\sqrt{\frac{1}{\Delta}\int_T^{T+\Delta}\xi_0(t)\eta_T(t)
 \exp\left(\frac{\nu^2C_H^2}{H}\left[(t-T)^{2H}-t^{2H}\right]\right)\D t}\right|\Ff_0\right]\\
 & \geq \EE\left[\left.\frac{1}{\Delta}\int_T^{T+\Delta}\sqrt{\xi_0(t)\eta_T(t)}
 \exp\left(\frac{\nu^2C_H^2}{2H}\left[(t-T)^{2H}-t^{2H}\right]\right)\D t\right|\Ff_0\right]\\
 & = \frac{1}{\Delta}\int_T^{T+\Delta}\sqrt{\xi_0(t)}\EE\left(\sqrt{\eta_T(t)}\right)
 \exp\left(\frac{\nu^2C_H^2}{2H}\left[(t-T)^{2H}-t^{2H}\right]\right)\D t\\
 & = \frac{1}{\Delta}\int_T^{T+\Delta}\sqrt{\xi_0(t)}
 \exp\left(\frac{\nu^2 C_{H}^2}{4H}\left[t^{2H}-(t-T)^{2H}\right]\right)
 \exp\left(\frac{\nu^2C_H^2}{2H}\left[(t-T)^{2H}-t^{2H}\right]\right)\D t\\
 & = \frac{1}{\Delta}\int_T^{T+\Delta}\sqrt{\xi_0(t)}
 \exp\left(\frac{\nu^2C_H^2}{4H}\left[(t-T)^{2H}-t^{2H}\right]\right)\D t,
\end{align*}
since $\eta_T(t)^{1/2}$ is log-normal (Proposition~\ref{prop:rBergomiVIX}), so that
$\EE(\sqrt{\eta_T(t)}) = \exp\left(\frac{\nu^2 C_{H}^2}{4H}\left[t^{2H}-(t-T)^{2H}\right]\right)$.
\end{proof}

We perform a numerical experiment to check the tightness of the bounds obtained in Proposition~\ref{thm:Bounds}). 
For this analysis we consider three qualitative scenarios for the initial forward variance curve:
\begin{equation}
\text{Scenario }1:\;\xi_0(t)=0.234^2;\quad
\text{Scenario }2:\;\xi_0(t)=0.234^2(1+t)^2; \quad
\text{Scenario }3:\;\xi_0(t)=0.234^2\sqrt{1+t}.\label{eq:scenarios}
\end{equation}
\begin{figure}[h]
\centering
\includegraphics[scale=0.2]{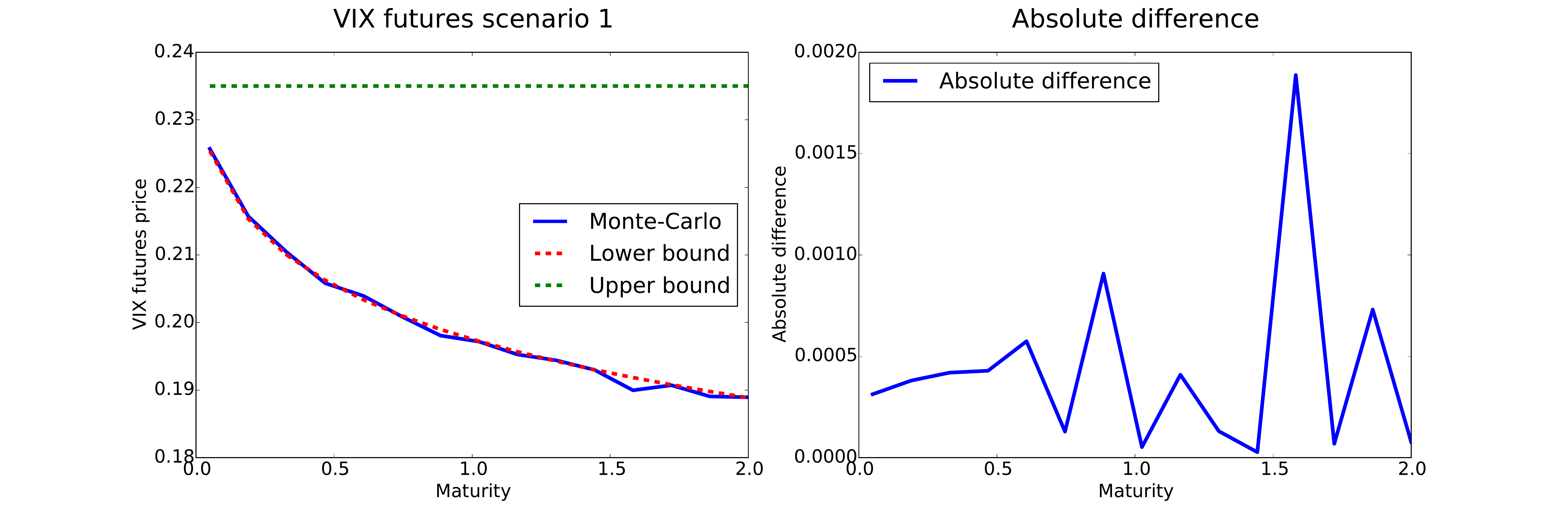}
\includegraphics[scale=0.2]{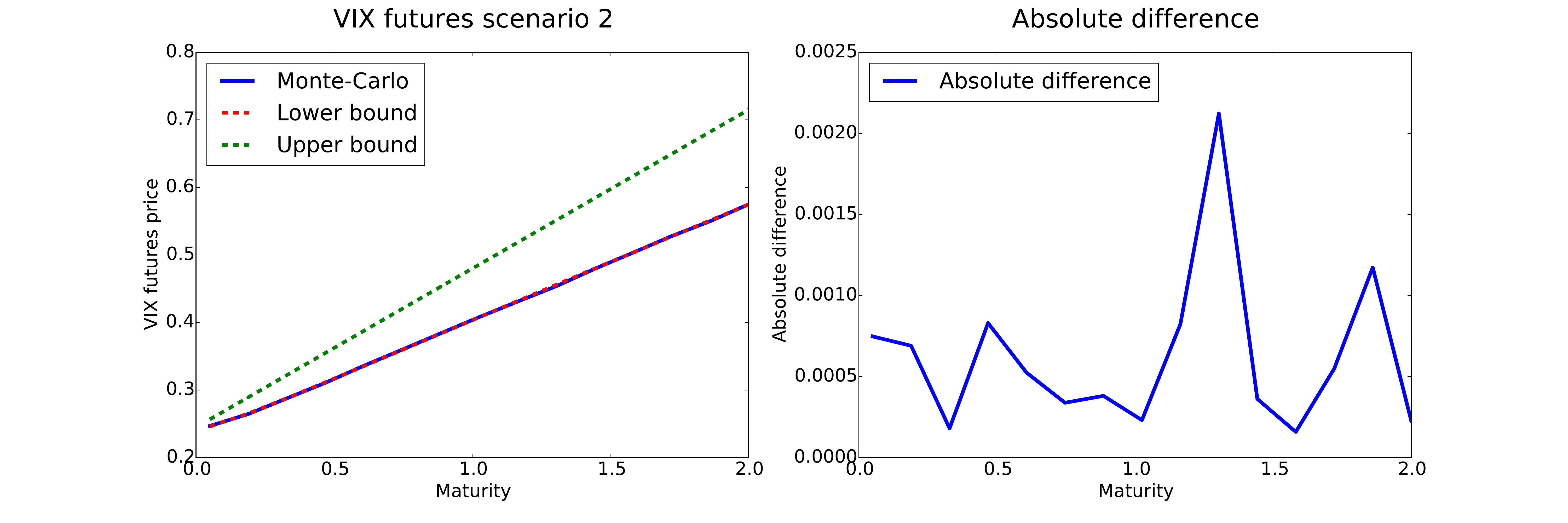}
\includegraphics[scale=0.2]{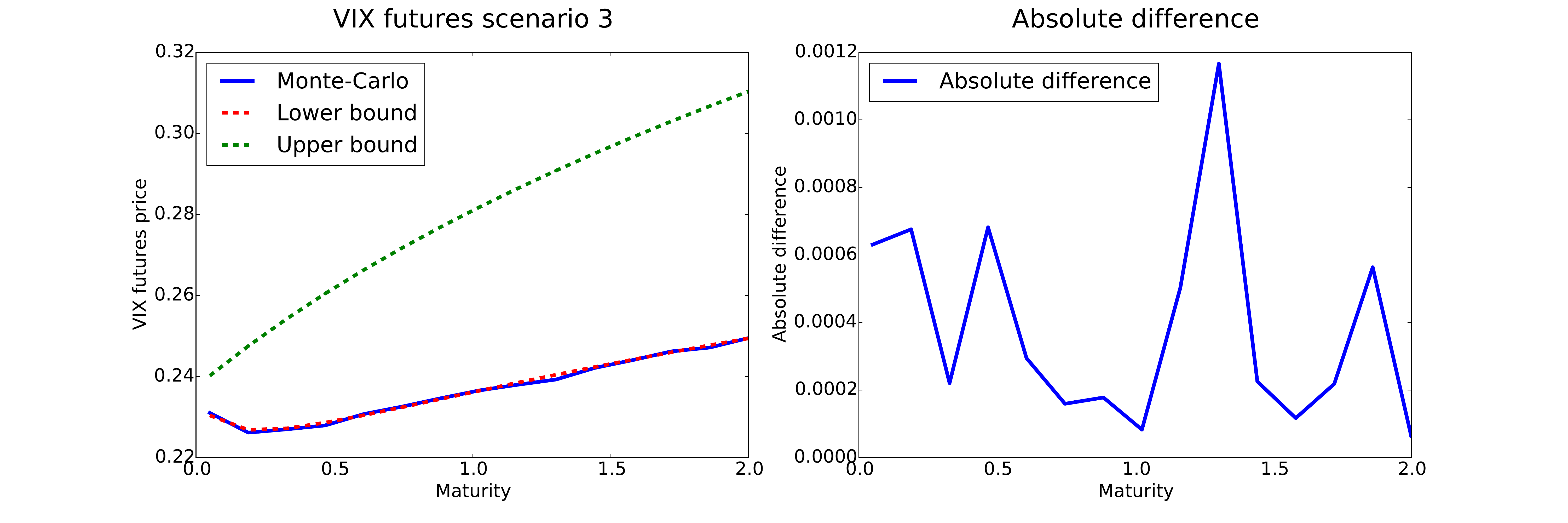}
\caption{Bounds vs. Monte Carlo (Truncated Cholesky) in all three scenarios.}
\label{fig:bounds}
\end{figure}
Figures~\ref{fig:bounds} suggest that the lower bound given in Proposition~\ref{thm:Bounds}) is surprisingly tight for very different shapes of~$\xi_0$. 
This can be explained with the following argument:
consider a simplified and deterministic version of the VIX futures price 
in Proposition~\ref{prop:rBergomiVIX}), denoted by
$$
\phi(T)
 := \sqrt{\frac{1}{\Delta}\int_T^{T+\Delta}f(t)\D t}
  = \sqrt{f(T)+\frac{\Delta}{2}f'(T)+\frac{\Delta^2}{6}f''(T)+\Oo(\Delta^3)},
$$
for some strictly positive (deterministic) function $f\in \Cc^2(\RR)$. 
We further introduce
$$
\psi(T)
 := \frac{1}{\Delta}\int_T^{T+\Delta}\sqrt{f(t)}\D t
  = \sqrt{f(T)}+\frac{\Delta}{4}\frac{f'(T)}{\sqrt{f(T)}}+\Oo(\Delta^2),
  $$
which is the corresponding lower bound by Cauchy-Schwarz's (or Jensen's) inequality, so that
$$
\phi(T)^2-\psi(T)^2=\Delta^2\left(\frac{f''(T)}{6}-\frac{f'(T)^2}{16f(T)}\right)+\Oo(\Delta^3).
$$
Hence, we observe that for small $\Delta$, as it is the case in VIX futures, the lower bound is very close from the original value which explains (at least for the deterministic case) the behaviour observed in Figure \ref{fig:bounds}.

\subsection{Numerical implementation of VIX process}
In this section, we investigate different simulation schemes for the VIX in the rough Bergomi model.

\subsubsection{Hybrid scheme and forward Euler approach}
In order to simulate the process $(\Vv_t^T)_{t\in[T,T+\Delta]}$ it is important to notice from~\eqref{eq:VtT} that the kernel has a singularity only for $\Vv^T_T$, hence we may overcome this by simulating the process using the hybrid scheme from Section~\ref{Hybrid simulation scheme}. Then, we may easily
extract~$Z$ and simulate~$(\Vv^T_t)_{t\in(T,T+\Delta]}$ using the forward Euler scheme with complexity $\Oo(n)$. A forward Euler scheme is chosen in this case due to the fact that the kernel in~\eqref{eq:VtT} no longer has a singularity for $t\in(T,T+\Delta]$ and this method is faster than the hybrid scheme.
Once~$(\Vv_t^T)_{t\in[T,T+\Delta]}$ is simulated, numerical integration routines may be used to simulate the VIX process 
using the expression in Proposition~\ref{prop:rBergomiVIX}.
It must be pointed out that this approach is computationally expensive and memory consuming
since it involves to simulate the Volterra process using the hybrid scheme with complexity 
of $\Oo(n\log n)$ and additionally, each~$\Vv^T_t$ by forward Euler.
The simulation algorithm can be summarised as follows:
\begin{algorithm}[VIX simulation in the rough Bergomi model]\label{algo:VIXrB}
Fix a grid $\Tt = \{t_i\}_{i=0,\ldots,n_T}$ and $\kappa\geq 1$.
\begin{enumerate}
\item Simulate the Volterra process~$(\Vv_t)_{t \in [0,T]}$ using the hybrid scheme 
in Appendix~\ref{sec:HybridScheme}, yielding a sample of the random variable $\Vv_T^T=\Vv_T$;
\item extract the path of the Brownian motion~$Z$ driving the Volterra process.
\begin{equation*}
\begin{array}{rll}
Z_{t_{i}}& =Z_{t_{i-1}}+n^{\Hm}\left(\Vv(t_{i}) - \Vv(t_{i-1})\right),& \text{for }i=1,\ldots,\kappa,\\
Z_{t_{i}}&=Z_{t_{i-1}}+\overline{Z}_{i-1} ,   & \text{for }i>\kappa;
  \end{array}
\end{equation*}
\item fix a grid $\mathfrak{T}=\{\tau_j\}_{j=0,...,N}$ on $[T,T+\Delta]$ and
approximate the continuous-time process~$\Vv^T$ by the discrete-time version~$\widetilde{\Vv}^T$ 
defined via the following forward Euler scheme:
$$
\widetilde{\Vv}^T_{\tau_0}  := \Vv^T_T
\qquad\text{and}\qquad
\widetilde{\Vv}^T_{\tau_j} :=  \sum_{i=1}^{n_T} \frac{Z_{t_i}-Z_{t_{i-1}}}{(\tau_j-t_{i-1})^{-\Hm}}, 
\qquad \text{for }j=1,\ldots,N;
$$
\item compute the VIX process via numerical integration, for example using a composite trapezoidal rule:
$$
\VIX_T \approx 
\left\{\frac{1}{\Delta}\sum_{j=0}^{N-1}\frac{Q^2_{T,\tau_{j}}+Q^2_{T,\tau_{j+1}}}{2}(\tau_j-\tau_{j-1})\right\}^{1/2},
$$
where 
$
Q^2_{T,\tau_j}
 := \xi_0(\tau_j)\exp\left(2v C_H \widetilde{\Vv}^T_{\tau_j}\right)\exp\left(\frac{\nu^2C_H^2}{H}\left((\tau_j-T)^{2H}-\tau_j^{2H}\right) \right)
$.
\end{enumerate}
\end{algorithm}
\begin{remark}
Step 4 may obviously be replaced by any available numerical integration routine,
but one must then carefully choose the partition in Step 3. 
\end{remark}

\subsubsection{Truncated Cholesky approach}
Alternatively, one could use the more expensive, yet exact, Cholesky decomposition to simulate~$\Vv^T$
on $[T,T+\Delta]$ since its covariance structure is known from~\eqref{eq:CovStructure}. 
However, computational complexity aside, with the same grid~$\Tt$ as in Algorithm~\ref{algo:VIXrB}, 
numerical experiments suggest that the determinant of the covariance matrix is equal to zero
when using more than $n_T = 8$ discretisation points. 
Hence, although valid in theory, the Cholesky approach is not feasible numerically. 
This fact implies that there exists strong linear dependence. 
In fact, for any $\eps>0$, the strict inequality
$\text{corr}(\Vv^T_{t_1},\Vv^T_{t_1+\eps})<\text{corr}(\Vv^T_{t_2},\Vv^T_{t_2+\eps})$
holds for all $T<t_1<t_2$, as well as the following:

\begin{proposition}\label{prop:CorrelEps}
The limit $\lim\limits_{\eps\downarrow 0}\text{corr}(\Vv^T_{t},\Vv^T_{t+\eps})=1$
holds for any $t\in[T,T+\Delta]$.
\end{proposition}
\begin{proof}
This follows readily from the continuity in~$\L^2$ of the map $t \mapsto \Vv^T_{t}$:
\begin{align*}
E[(\Vv^T_{t+\eps} - \Vv^T_{t})^2] &= \int_0^T [(t+\eps-u)^{\Hm}- (t-u)^{\Hm}]^2 \D u\\
 & = \left[\frac{(T+\eps)^{1+2\Hp}-\eps^{1+2\Hp}+T^{1+2\Hp}}{1+2\Hp}
  - \frac{2T^{1+\Hp}\eps^{\Hp}}{1+\Hp}\text{}_2F_1\left(-\Hp,1+\Hp,2+\Hp,-\frac{T}{\eps}\right)\right].
\end{align*}
Applying the identity~\cite[page 564]{Abramowitz-Stegun}
$$
\text{}_2F_1(a,b,c;z)
 = \frac{\Gamma(c)\Gamma(b-a)}{\Gamma(b)\Gamma(c-a)}(-z)^{-a}\text{}_2F_1\left(a,a-c+1,a-b+1;\frac{1}{z}\right) 
 + \frac{\Gamma(c)\Gamma(a-b)}{\Gamma(a)\Gamma(c-b)}\frac{1}{(-z)^{b}}\text{}_2F_1\left(b,b-c+1,b-a+1;\frac{1}{z}\right)
$$
to the case $(a, b, c) = (-\Hp,1+\Hp,2+\Hp)$ and using properties of the Gamma function along with the fact that $\text{}_2F_1\left(a,0,c,z\right)\equiv 1$, we obtain
$$
\text{}_2F_1\left(-\Hp,1+\Hp,2+\Hp,-\frac{T}{\eps}\right)
 = \frac{1+\Hp}{1+2\Hp}\left(\frac{T}{\eps}\right)^{\Hp}\text{}_2F_1\left(-\Hp,2+2\Hp,-2\Hp,\frac{-\eps}{T}\right).
 $$
Finally, the series representation of~$\text{}_2F_1$~\cite[Chapter 15.1.1]{Abramowitz-Stegun} implies that 
$\text{}_2F_1\left(-\Hp,2+2\Hp,-2\Hp,\frac{-\eps}{T}\right)$ converges to~$1$ as~$\eps$ tends to zero,
and hence that $E[(\Vv^T_{t+\eps} - \Vv^T_{t})^2]$ tends to zero.
\end{proof}

In light of Proposition~\ref{prop:CorrelEps}, 
we model exactly the dependence structure on the first 8 grid points $t_1,\ldots,t_8$, 
then truncate the Cholesky decomposition and compute the correlations
$\rho_{i}:=\text{corr}(\Vv^T_{t_{8+i}},\Vv^T_{t_{8+i+1}})$, for $i=0,\ldots,n-9$ 
to approximate the process by adequately rescaling and correlating each pair of subsequent grid points.
In contrast to the hybrid + forward Euler scheme, the computational complexity is much lower, 
since the Cholesky method is truncated with only 8 components.
The VIX simulation algorithm therefore reads as follows:
\begin{algorithm}[VIX simulation (truncated Cholesky)]
Fix a grid $\mathfrak{T}=\{\tau_j\}_{j=0,...,N}$ on $[T,T+\Delta]$,
\begin{enumerate}[(i)]
\item compute the covariance matrix of  $(\Vv_{\tau_j}^T)_{i=j,\ldots,8}$
 using~\eqref{eq:CovStructure};
\item generate $\{\Vv^T_{\tau_j}\}_{j=1,...,N}$  by correlating and rescaling using~\eqref{eq:CovStructure}:
$$
\Vv^T_{\tau_j}
 = \sqrt{\VV(\Vv^T_{\tau_j})}
 \left(\frac{\rho(\Vv^T_{\tau_{j-1}},\Vv^T_{\tau_j})\Vv^T_{\tau_{j-1}}}{\sqrt{\VV(\Vv^T_{\tau_{j-1}})}}+\sqrt{1-\rho(\Vv^T_{\tau_{j-1}},\Vv^T_{\tau_j})^2}\Nn(0,1)\right),
 \quad\text{for }j=9,\ldots,N;
$$
\item compute the VIX via numerical integration as in Algorithm~\ref{algo:VIXrB}(4).
\end{enumerate}
\end{algorithm}

\subsubsection{Numerical experiment}
We compute the price of VIX Futures using the simulation algorithm 
introduced in the previous section. 
We set the same parameters as in~\cite{BFG15} and~\cite{BLP15}:
\begin{equation}\label{eq:parameters}
\xi_0 = 0.235^2,\qquad\qquad
H = 0.07,\qquad\qquad
\nu = 1.9\frac{C_H\sqrt{2H}}{2}\approx 1.2287,\qquad\qquad
\kappa = 2.
\end{equation}
We perform $10^5$ simulations for the hybrid scheme + forward Euler (HS+FE) method, 
while $10^6 $ simulations are used for the truncated Cholesky.
Figure~\ref{sim VIX} suggests that both methods agree qualitatively and converge to a similar output. 
In particular, the truncated Cholesky approach seems to suffer from larger oscillations as maturity increases, 
even with~$10^6$ simulations. 
Nevertheless, Figure~\ref{sim VIX} indicates that 
the Monte-Carlo variance increases in~$T$ for both schemes, 
which is confirmed in Figure~\ref{MC_maturities}, where the error also increases with maturity. 
On the other hand, Figure~\ref{MC_error} shows that the HS+FE methods is slower 
than the truncated Cholesky method, 
which is consistent with the computational complexities discussed in the previous section. 
In particular, the computational time of the Cholesky method is almost constant when using parallel computing. 
Figure~\ref{MC_error} also suggests that large simulations are needed to obtain precise prices.
In light of this analysis, both methods seem to approximate 
the required output in a decent manner. 
Even if the truncated Cholesky approach gives a considerably fast output for each maturity, 
it is not considered for calibration, since its computational time grows linearly in the number of maturities, 
making the algorithm too slow for reasonable calibration.
Instead, we will use the truncated Cholesky approach as a benchmark for the upcoming approximations.

\begin{figure}[h]
\centering
\includegraphics[scale=0.2]{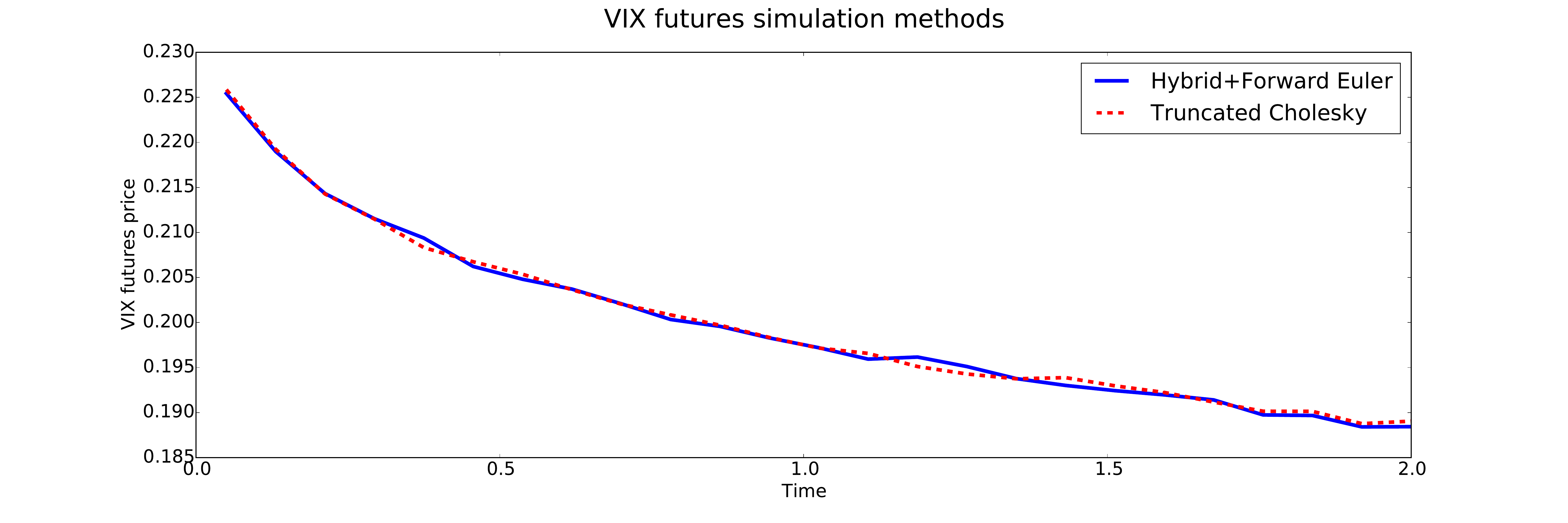}
\caption{VIX Futures using HS+FE and truncated Cholesky}
\label{sim VIX}
\end{figure}

\begin{figure}[h]
\centering
\includegraphics[scale=0.2]{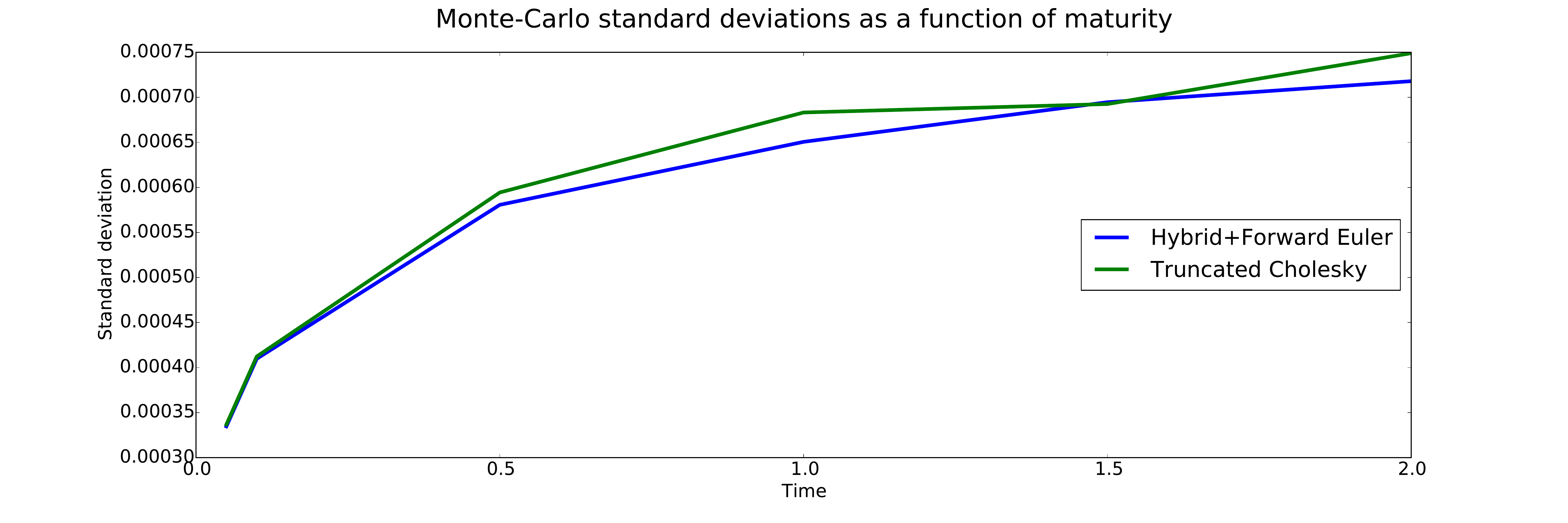}
\caption{Monte-Carlo standard deviations, with $10^5$ simulations.}
\label{MC_maturities}
\end{figure}

\begin{figure}[h]
\centering
\includegraphics[scale=0.2]{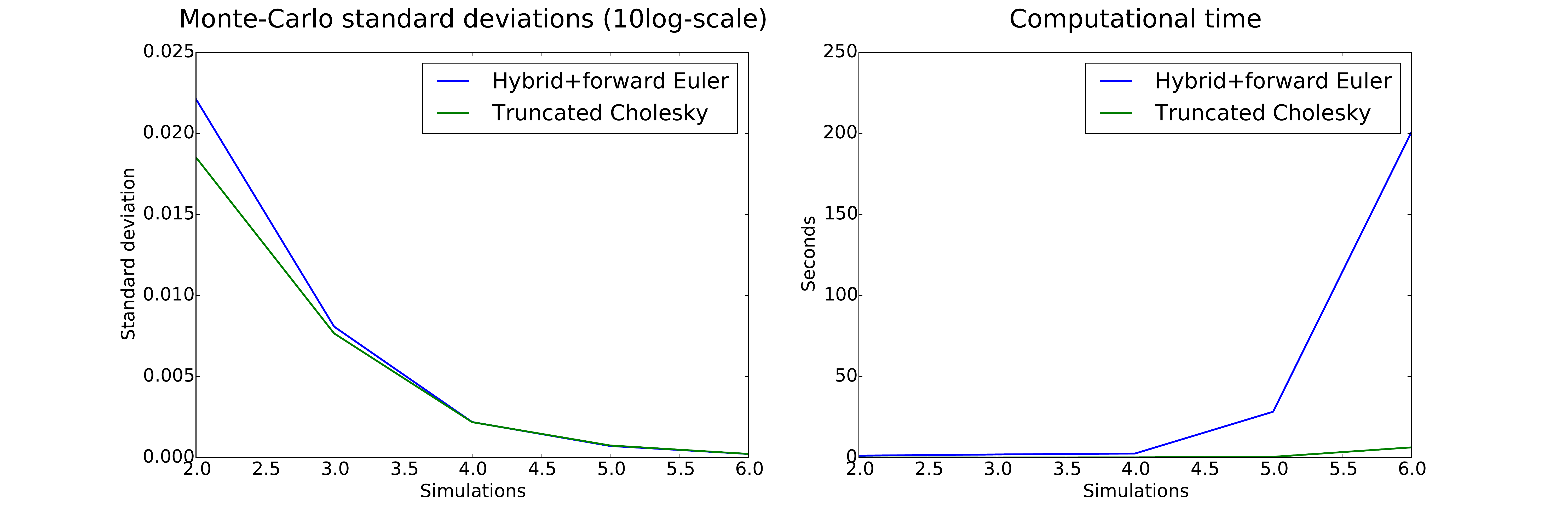}
\caption{Monte-Carlo standard deviations and computational times for a fixed maturity ($T=2$ years) 
for both methods using efficient parallel computing.}
\label{MC_error}
\end{figure}

\subsection{The VIX process and log-normal approximations}\label{sec:LNVIX}
We now investigate approximate methods to price VIX futures and options. 
We define the $\Ff_T$-measurable random variable
$\Delta\VIX^2_T = \int_T^{T+\Delta}\xi_T(t)\D t$.
In~\cite{BFG15}, Bayer, Friz and Gatheral assumed that $\log(\Delta\VIX^2_T)$ follows a Gaussian distribution,
and computed directly its first and second moments, which hence fully characterise the distribution of~$\Delta\VIX^2_T$.
However, since the sum of log-normal random variables is not (in general) log-normal, 
the fact that~$\xi_T$ is log-normal does not imply that~$\Delta\VIX^2_T$ is.
In a different context (geometric Brownian motion and Asian options), 
Dufresne~\cite{Dufresne} proved that, under certain conditions,
an integral of log-normal variables asymptotically converges to a log-normal. 
This approximation has been widely used for many applications~\cite{Dufresne}, 
and Dufresne's result motivates Bayer-Friz-Gatheral's assumption. 
We provide here exact formulae for the mean and variance of this distribution, 
and compare them numerically to those by Bayer-Friz-Gatheral.

\begin{proposition}\label{prop:VIXLN}
The following holds:
\begin{align*}
\sigma^2 & :=\VV(\log(\Delta \VIX^2_T)) =  - 2\log\EE(\Delta \VIX^2_T) + \log\EE\left((\Delta \VIX^2_T)^2\right),\\
\mu  & := \EE(\log(\Delta \VIX^2_T)) = \log\EE(\Delta \VIX^2_T)-\frac{\sigma^2}{2}.
\end{align*}
Furthermore,  
$\EE(\Delta \VIX^2_T) = \int_T^{T+\Delta}\xi_0(t)\D t$ and 
\begin{equation}\label{eq:variance}
\EE\left[(\Delta \VIX^2_T)^2\right]
 = \int_{[T, T+\Delta]^2}\xi_0(u)\xi_0(t)
\exp\left\{\frac{\nu^2C_H^2}{H}\left[(u-T)^{2H}+(t-T)^{2H}-u^{2H}-t^{2H}\right]\right\}
\E^{\overline{\Theta}_{u,t}}\D u \D t.
\end{equation}
where $\overline{\Theta}_{u,t}$ is equal to zero if $u=t$ and otherwise equal to
$\Theta_{u\vee t,u \wedge t}$, where
\begin{align*}
\Theta_{u,t}
  := 2\nu^2C^2_H\left\{\frac{u^{2H}-(u-T)^{2H}+t^{2H}-(t-T)^{2H}}{2H}
  + 2\frac{(u - t)^{\Hm}}{\Hp}
  \left[ t^{\Hp}\F\left(\frac{-t}{u - t}\right) -(t - T)^{\Hp} \F\left(\frac{T-t}{u - t}\right)\right]\right\}.
\end{align*}
\end{proposition}

\begin{remark}
Since all the integrals in the proposition are computed over compact intervals, 
they are finite as long as~$\xi_0$ is well behaved on $[T, T+\Delta]$.
Assuming this is indeed the case is not restrictive in practice as~$\xi_0$ represents the initial forward variance curve; 
note in particular that continuity of~$\xi_0$ is sufficient.
\end{remark}

\begin{remark}
The reason why $\overline{\Theta}_{u,t}$ is defined that way is for numerical purposes.
Indeed, when $u<t$, $\Theta_{u,t}$ is not well defined since~$\F(x)$ only makes sense when $x\leq1$ 
(details on the radius of convergence of hypergeometric functions can be found in~\cite{Abramowitz-Stegun}[Page 556]), even though the integral representation~\eqref{eq:variance} is still well defined. 
However, most numerical packages implement hypergeometric functions via series expansions.
The trick from $\Theta_{u,t}$ to $\overline{\Theta}_{u,t}$ allows us to bypass this issue.
\end{remark}

\begin{proof}[Proof of Proposition~\ref{prop:VIXLN}]
The expectation follows directly from the tower property and Fubini's theorem.
For the second moment, we use the decomposition
$$
\EE\left((\Delta \VIX_T^2)^2\right)
 = \EE\left(\int_T^{T+\Delta}\int_T^{T+\Delta}\xi_T(u)\xi_T(t)\D t \D u\right)
 = \int_T^{T+\Delta}\int_T^{T+\Delta}\EE\left(\xi_T(u)\xi_T(t)\right)\D t \D u
$$
where in the last step we use that $\xi_T(s)$ is $\mathcal{F}_T$-measurable in order to apply Fubini.
Using the representation obtained in Proposition~\ref{prop:xiMartingale}, we get
\begin{equation}\label{expect}
\EE\left(\xi_T(u)\xi_T(t)\right)
 =  \xi_0(u)\xi_0(t)\exp\left\{\frac{\nu^2C_H^2}{H}
 \left[ (u-T)^{2H}-u^{2H} + (t-T)^{2H}-t^{2H}\right]\right\}\EE\left(\E^{\vartheta_{u,t}}\right),
\end{equation}
where the random variable
$$
\vartheta_{u,t} := 2\nu C_{H}\int_{0}^{T}\left[(u-s)^{\Hm} + (t-s)^{\Hm} \right]\D Z_s
$$
is Gaussian with zero expectation and variance 
\begin{align*}
\VV(\vartheta_{u,t})
 & = 4\nu^2C^2_H\left(\frac{u^{2H}-(u-T)^{2H}+t^{2H}-(t-T)^{2H}}{2H}
  + 2\int_{0}^{T}(u-s)^{\Hm}(t-s)^{\Hm}\D s  \right)\\
 & = 4\nu^2C^2_H\left\{\frac{u^{2H}-(u-T)^{2H}+t^{2H}-(t-T)^{2H}}{2H}
  + \frac{2(u - t)^{\Hm}}{\Hp}
  \left[t^{\Hp}\F\left(\frac{-t}{u - t}\right) - (t - T)^{\Hp} \F\left(\frac{T-t}{u - t}\right)\right]\right\}.
\end{align*}
Then,
\begin{align*}
\EE\left(\xi_T(u)\xi_T(t)\right)
 = \xi_0(u)\xi_0(t)\exp\left(\frac{\nu^2C_H^2}{H}\left((u-T)^{2H}+(t-T)^{2H}-u^{2H}-t^{2H}\right)\right)
\exp\left(\frac{\VV(\vartheta_{u,t})}{2}\right),
\end{align*}
and the second moment follows from the immediate computation
\begin{align*}
\EE\left((\Delta \VIX_T)^2\right)
 & = \int_T^{T+\Delta}\int_T^{T+\Delta}\EE\left[\xi_T(u)\xi_T(t)\right]\D u \D t\\
 & =\int_T^{T+\Delta}\int_T^{T+\Delta} \xi_0(u)\xi_0(t)
\exp\left\{\frac{\nu^2C_H^2}{H}\left[(u-T)^{2H}+(t-T)^{2H}-u^{2H}-t^{2H}\right]\right\}
\E^{\frac{1}{2}\VV(\vartheta_{u,t})}\D u \D t,
\end{align*}
after defining $\Theta_{u,t} := \exp\left(\frac{1}{2}\VV(\vartheta_{u,t})\right)$.
\end{proof}

In order to provide closed-form expressions for VIX Futures and options, 
we enforce the following assumption:
\begin{assumption}\label{ass:VIXLN}
$\Delta \VIX^2_T$ is log-normal.
\end{assumption}
In fact, this is almost the same as the assumption by Bayer, Friz and Gatheral~\cite{BFG15};
however, they did not compute the variance exactly as in Proposition~\ref{prop:VIXLN}, 
and instead considered the lognormal approximation
\begin{assumption}[Bayer-Friz-Gatheral]\label{ass:Bayer}
$\log(\Delta \VIX^2_T)$ is Gaussian with mean~$\tmu$ and variance~$\tsigma^2$ given by
$$
\tsigma^2 = \frac{4\nu^2C_H^2}{\Delta^2\Hp^2}\int_0^T\left[(T-s+\Delta)^{\Hp}-(T-s)^{\Hp}\right]^{2}\D s
\qquad\text{and}\qquad
\tmu=\log \int_T^{T+\Delta}\xi_0(t)\D t -\frac{\tsigma^2}{2}.
$$
\end{assumption}

\begin{lemma}
A VIX future is worth 
\begin{equation*}
\Vk_T  = \left\{
 \begin{array}{ll}
 \displaystyle \Delta^{-1/2}\sqrt{\int_T^{T+\Delta}\xi_0(t)\D t} \exp\left(-\frac{\sigma^2}{8}\right),
  & \text{under Assumption~\ref{ass:VIXLN}},\\
 \displaystyle \Delta^{-1/2}\sqrt{\int_T^{T+\Delta}\xi_0(t)\D t} \exp\left(-\frac{\tsigma^2}{8}\right),
  & \text{under Assumption~\ref{ass:Bayer}}.
\end{array}
\right.
\end{equation*}
\end{lemma}
\begin{proof}
Since, by assumption
$\Delta^{1/2}\VIX_T\equalDistrib \exp\left(\frac{\mu+\sigma \Nn(0,1)}{2}\right)$, 
the price of a VIX future directly reads
$$
 \EE\left(\VIX_T|\Ff_0\right)
  = \Delta^{-1/2}\EE\left(\Delta^{1/2}\VIX_T\right)
  = \Delta^{-1/2}\exp\left(\frac{\mu}{2}+\frac{\sigma^2}{8}\right),
$$
and the second case follows analogously.
\end{proof}
As opposed to the simulation schemes, 
the log-normal approximation does depend on the qualitative properties 
of the initial forward variance curve~$\xi_0(\cdot)$. 
Hence, different curves should be analysed to check the robustness of the method.
We now exploit the approximation in Assumption~\ref{ass:VIXLN} 
to obtain a closed-form Black-Scholes type formulae for European options on the VIX. 
\begin{lemma}
For $0\leq t\leq T$, let $\Vk_T(t) := \EE\left(\VIX_T|\Ff_t\right)$ denote the price at time~$t$ 
of a VIX future maturing at~$T$. 
Then, under Assumption~\ref{ass:VIXLN}, a European Call option on a VIX future maturing at $T$ is worth
$$
\Cc^\Vv_T := \EE\left[(\Vk_T(T)-K)_+|\Ff_0\right]
 = \Delta^{-1/2}\sqrt{\int_T^{T+\Delta}\xi_0(t)\D t}\exp\left(-\frac{\sigma^2}{8}\right)\Phi(d_1)-K\Phi(d_2),
$$
where $\widetilde{K} := [\log (K^2\Delta)-\log\int_T^{T+\Delta}\xi_0(t)\D t+\sigma^2/2]/\sigma$, 
$d_1 := -\widetilde{K}+\frac{1}{2}\sigma$
and $d_2 := -\widetilde{K}$, with~$\sigma^2$ as in Proposition~\ref{prop:VIXLN}
\end{lemma}
\begin{proof}
Under Assumption~\ref{ass:VIXLN}, the lemma follows directly from the following trivial computations:
$$
\EE\left(\Vk_T(T)-K\right)_+
 = \int_{\widetilde{K}}^{\infty}\left[\Delta^{-1/2}\exp\left(\frac{\mu}{2}+\frac{\sigma z}{2}\right)-K\right]_+\phi(z)\D z
 = \frac{1}{\sqrt{\Delta}}\exp\left(\frac{\mu}{2}+\frac{\sigma^2}{8}\right)\int_{\widetilde{K}}^{\infty}\phi\left(z-\frac{\sigma}{2}\right)\D z-K\Phi(d_2).
 $$
\end{proof}

\subsubsection{Numerical tests of the log-normal approximation}
We perform a numerical analysis of the approximation by pricing VIX Futures using the parameters 
in~\eqref{eq:parameters}. 
We consider $10^6$ simulations and the three qualitative scenarios introduced in \eqref{eq:scenarios} for the forward variance curve.
\begin{figure}[h]
\centering
\includegraphics[scale=0.2]{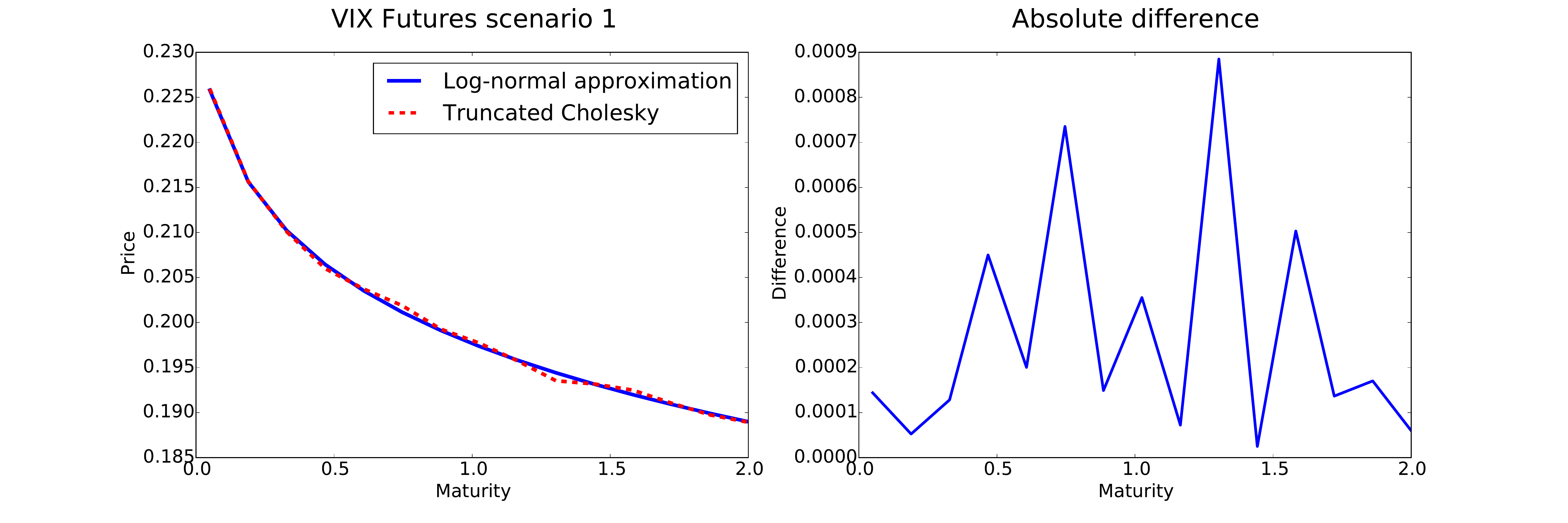}
\includegraphics[scale=0.2]{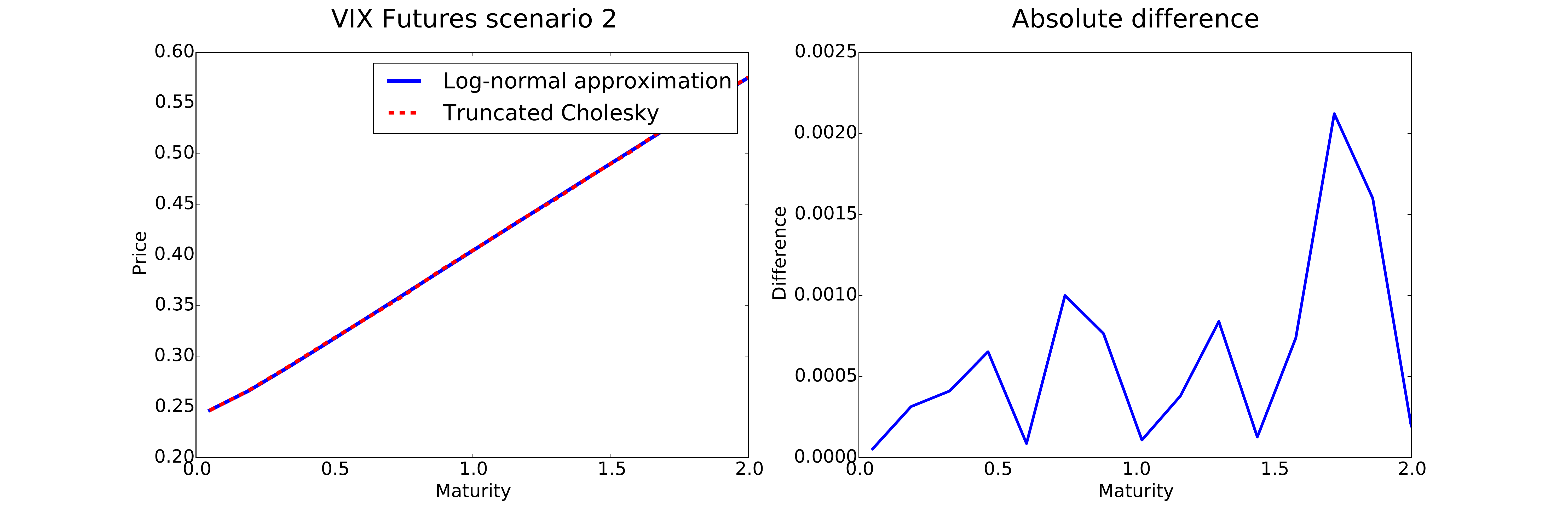}
\includegraphics[scale=0.2]{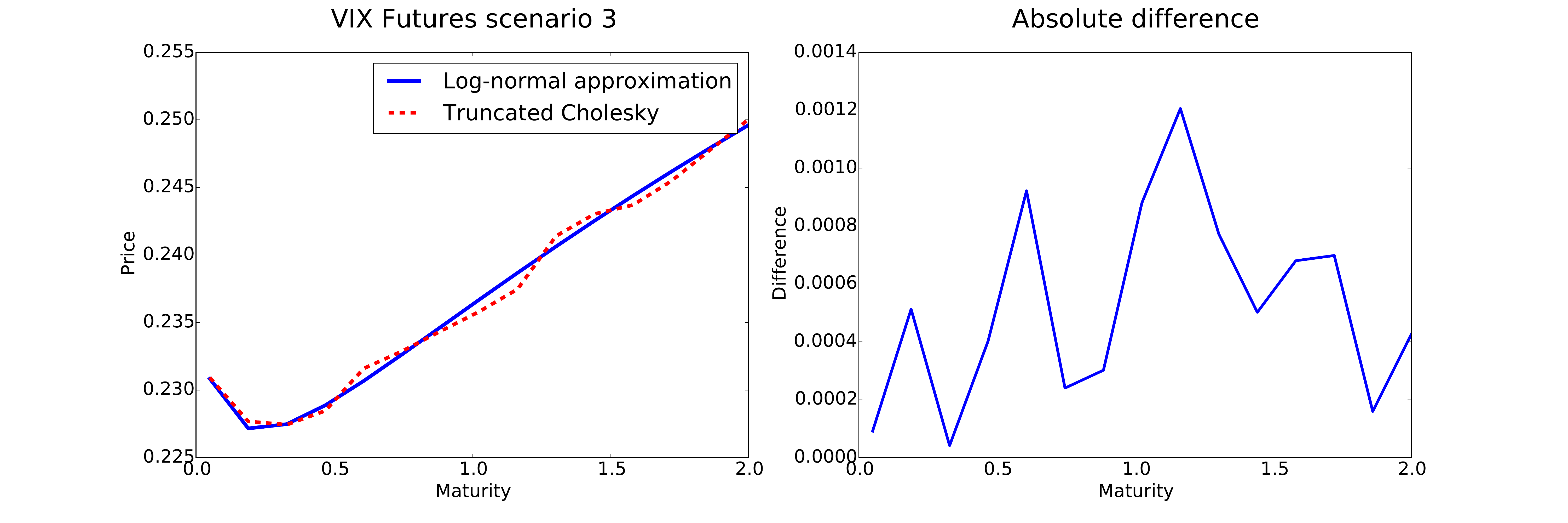}
\caption{Log-normal approximations vs. Monte Carlo (Truncated Cholesky) in all three scenarios.}
\label{fig:log_approx}
\end{figure}
Figures~\ref{fig:log_approx} suggest that the approximation is accurate being the difference of  order $10^{-3}$ or less. The oscillating nature of the difference is also a good sign since it is probably caused by the Monte Carlo error and does not show any monotonicity. Furthermore, the method shows to be robust for different type of curves $\xi_0(\cdot)$. Moreover, the log-normal approximation seems to converge to the true mean, avoiding the oscillations of the Truncated Cholesky method. Therefore, we may conclude that the log-normal approximation produces a good output, with desired smoothness properties.
On the other hand, we also analyse European VIX Call options repeating the previous parameters and considering the flat forward variance curve from Scenario~$1$.
\begin{figure}[h]
\centering
\includegraphics[scale=0.2]{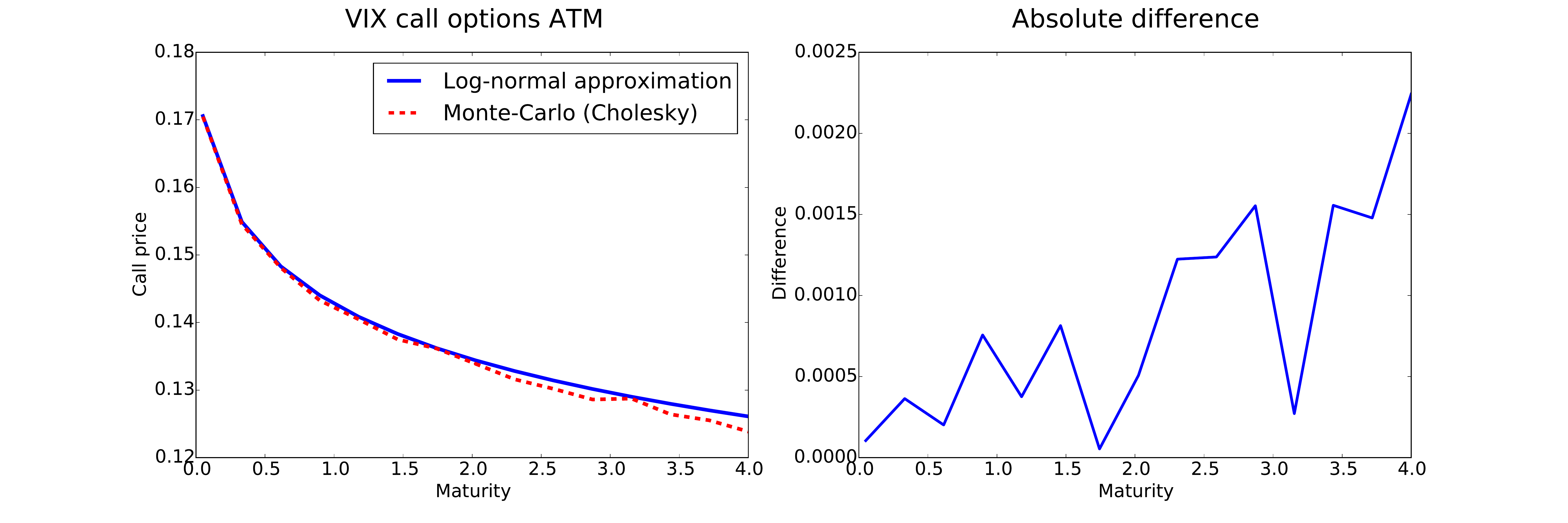}

\caption{Log-normal approximation vs. Monte Carlo (Truncated Cholesky) with $10^6$ simulations}
\label{fig:log_options_ATM}
\end{figure}
Figure ~\ref{fig:log_options_ATM} shows that the log-normal approximation is accurate for at-the-money options. 
Nevertheless, the difference increases with maturity, and hence one must be careful 
when using this approximation to price options with long maturities.
However, in practice, on Equity markets, liquid maturities are only up to two years.

\subsubsection{Numerical tests}
We benchmark Bayer-Friz-Gatheral's approximation against ours with the parameters 
in~\eqref{eq:parameters}. 
Figures~\ref{our_vs_BFG} and~\ref{our_vs_BFG_variance} suggest that both are similar. 
In particular, we observe that the variance deviates as maturity increases. 
Nevertheless, for practical purposes, both approximations agree over a four-year horizon, 
long enough to cover available Futures data.
Many functional forms of~$\xi_0(t)$ were also tested giving similar results as the ones shown 
in Figures~\ref{our_vs_BFG} and~\ref{our_vs_BFG_variance}. 
We conclude that the approximation by Bayer, Friz and Gatheral is accurate enough for practical purposes
and yields and reduces significantly the computational costs 
since the computation of~$\tsigma^2$ involves a single integral 
while our approximation requires a double integral (for~$\sigma^2$). 
In particular, in order to generate Figure~\ref{our_vs_BFG} our approximation was 40 times slower than Bayer, Friz and Gatheral's.

\begin{figure}[h]
\centering
\includegraphics[scale=0.2]{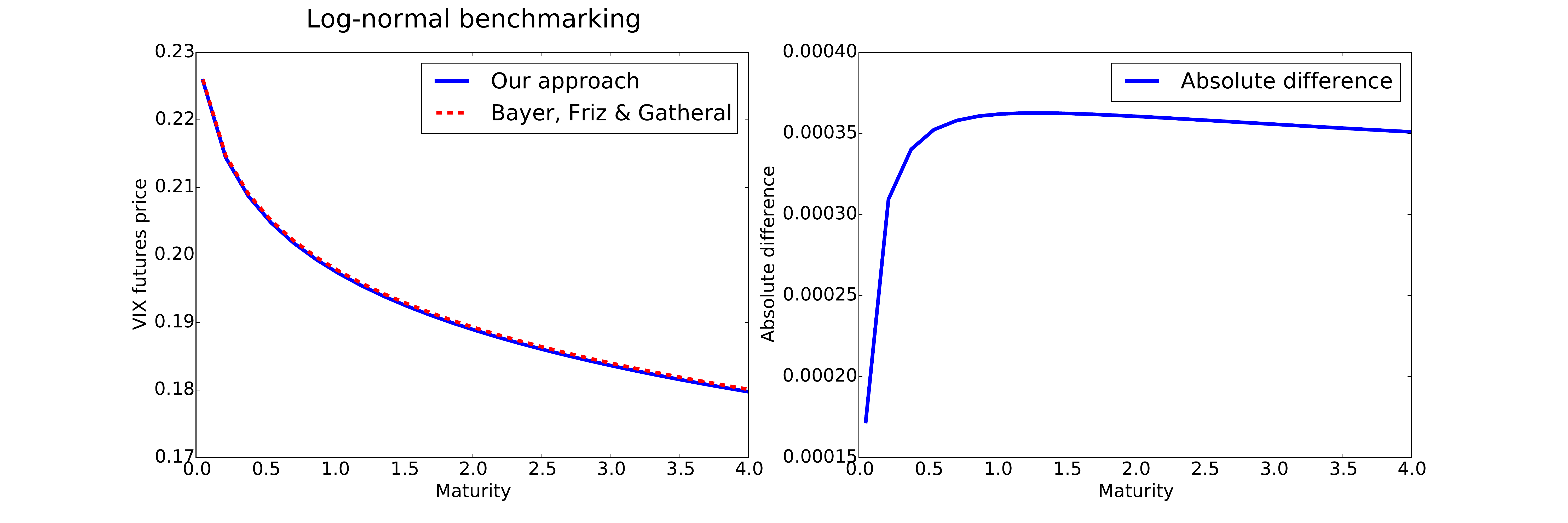}

\caption{VIX Futures prices following Assumption~\ref{ass:VIXLN} and Assumption~\ref{ass:Bayer}.}
\label{our_vs_BFG}
\end{figure}
\begin{figure}[h]
\centering
\includegraphics[scale=0.2]{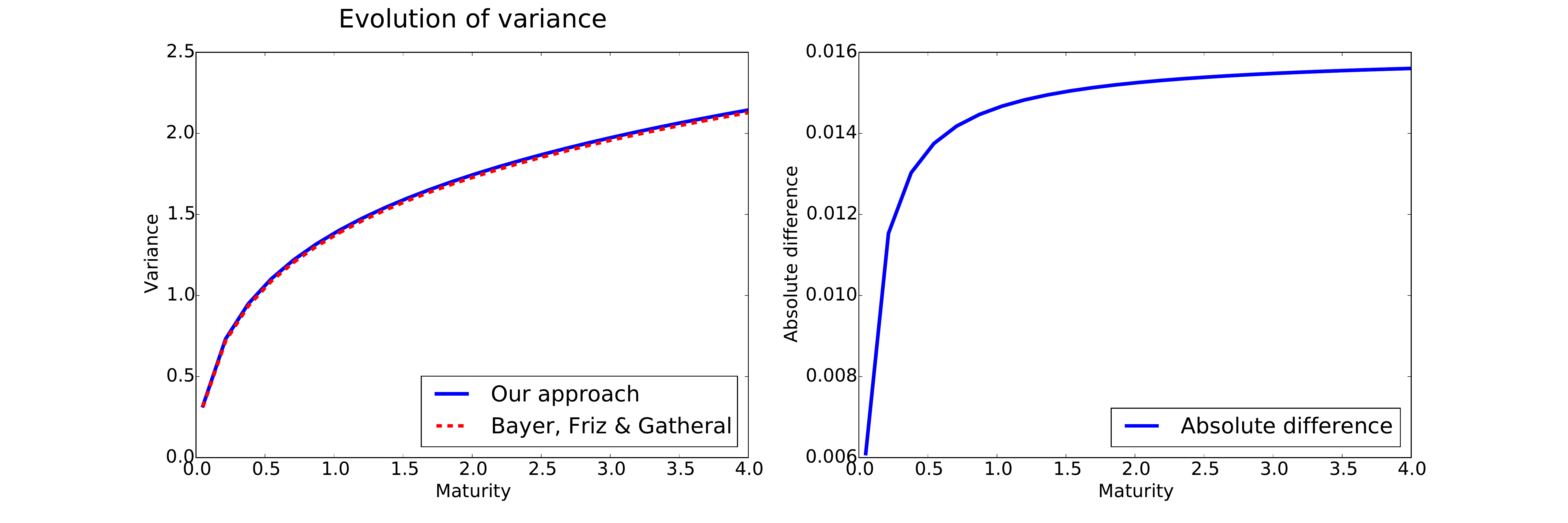}

\caption{Comparison of the variance following Assumption~\ref{ass:VIXLN} ($\sigma^2$) and Assumption~\ref{ass:Bayer} ($\tsigma^2$).}
\label{our_vs_BFG_variance}
\end{figure}

\subsection{VIX Futures calibration in rBergormi}
In light of the promising results obtained in the previous sections, we create a calibration algorithm based on the log-normal approximation by Bayer, Friz and Gatheral~\cite{BFG15}. Even if our approximation seems to be more accurate the computational cost is much larger and the difference between both approaches has been shown to be very small. Moreover, the approach by Bayer, Friz and Gatheral allows to compute the gradient of the objective function in a semi-explicit form, which in terms of optimisation and calibration is extremely useful.

\subsubsection{Objective function}
To calibrate the model to VIX Futures, we first define the objective function
\begin{equation}\label{eq:ObjectFunc}
\Ll^{\Fk}(\nu,H) := \sum_{i=1}^{N} (\Vk_{T_i}-\Fk_i)^2,
\end{equation}
which we minimise over $(\nu, H)$.
Here $(\Fk_i)_{i=1,\ldots,N}$ are the observed Futures prices on the time grid $T_1<\ldots<T_N$,
$\Vk_{T_i} = \Delta^{-1/2}\sqrt{\int_{T_i}^{T_i+\Delta}\xi_0(t)\D t}\exp\left(- \frac{\tsigma_i^2}{8}\right)$.
and
$$
\tsigma_i^2 = \frac{4\nu^2 C_H^2}{\Hp^2\Delta^2}\left[\frac{(T_i+\Delta)^{1+2\Hp}-\Delta^{1+2\Hp}+T_i^{1+2\Hp}}{1+2\Hp}-2\frac{T_i^{1+\Hp}\Delta^{\Hp}}{1+\Hp}\text{}_2F_1\left(-\Hp,1+\Hp,2+\Hp,-\frac{T_i}{\Delta}\right)\right],
$$
which is the closed-form expression of the variance in Assumption~\ref{ass:Bayer}.
The gradient of the objective function is an important source of information in many optimisation algorithms. 
To compute it, we differentiate the objective function in~\eqref{eq:ObjectFunc} 
with respect to~$\nu$ and~$H$ and apply the chain rule:
$$
\frac{\partial \Ll^\Fk}{\partial \nu}(\nu,H)
 = -\frac{1}{4}\sum_{i=1}^{N} (\Vk_{T_i}-\Fk_i)\Vk_{T_i}\frac{\partial \tsigma_i^2}{\partial\nu}
 = -\frac{1}{2\nu}\sum_{i=1}^{N} (\Vk_{T_i}-\Fk_i)\Vk_{T_i} \tsigma_i^2,
$$
where 
$$
\frac{\partial \tsigma_i^2}{\partial\nu}
 = \frac{8\nu C_H^2}{\Delta^2\Hp^2}\int_0^{T_i} \left((T_i-s+\Delta)^{\Hp}-(T_i-s)^{\Hp}\right)^{2}\D s
 = \frac{2\tsigma_i^2}{\nu}.
 $$
On the other hand,
$\displaystyle \frac{\partial \Ll^\Fk}{\partial H}(\nu,H)
 = -\frac{1}{4}\sum_{i=1}^{N} (\Vk_{T_i}-\Fk_i)\Vk_{T_i}\frac{\partial \tsigma_i^2}{\partial H}$,
with
\begin{align*}
\frac{\partial \tsigma_i^2}{\partial H}
 & = \frac{4\nu^2\frac{\partial C_H^2}{\partial H}\Hp-8\nu^2C_H^2}{\Delta^2\Hp^3}\int_0^{T_i}
 \left((T_i-s+\Delta)^{\Hp}-(T_i-s)^{\Hp}\right)^{2}\D s\\
 & + \frac{8\nu^2C_H^2}{\Delta^2\Hp^2}\int_0^{T_i} \left[(T_i-s+\Delta)^{H+\frac{1}{2}}-(T_i-s)^{\Hp}\right]
\left[(T_i-s+\Delta)^{\Hp}\log(T_i-s+\Delta)-(T_i-s)^{\Hp}\log(T_i-s)\right]\D s\\
 & = \frac{\tsigma_i^2(K_H-2)}{\Hp}\\
 &\quad  + \frac{8\nu^2C_H^2}{\Delta^2\Hp^2}
 \int_0^{T_i}\left[(T_i-s+\Delta)^{\Hp}-(T_i-s)^{\Hp}\right]
\left[(T_i-s+\Delta)^{\Hp}\log(T_i-s+\Delta)-(T_i-s)^{\Hp}\log(T_i-s)\right]\D s,
\end{align*}
where
\begin{align*}
\frac{\partial C_H^2}{\partial H}
 = \frac{2\Gamma(2-\Hp)}{\Gamma(\Hp)\Gamma(1-2\Hm)}
 \Big\{ 1+\left[2\psi(2-\Hp) - \psi(\Hp) - \psi(1-2\Hm)\right]H\Big\}
\end{align*}
and
$K_H=\frac{\Hp}{H}\Big\{1+H\psi\left(2-\Hp\right)-H\left(\psi(\Hp)+\psi(1-2\Hm)\right)\Big\}$,
where $\psi$ is the digamma function.

\subsubsection{Obtaining the initial forward variance curve}\label{sec:InitialFwdCurve}
The initial forward variance curve plays a crucial role in both the Bergomi~\cite{Bergomi} 
and the rough Bergomi models~\cite{BFG15}, since it is a market input. 
In particular it depends on the current term structure of variance swaps. 
Even if variance swaps are not traded in standard exchange markets, they are actively traded over-the-counter (OTC). 
This in turn means that there is no observable data and we must establish a valuation method for variance swaps. 
The celebrated static replication formula by Carr and Madan~\cite{Carr-Madan},
applicable here since the underlying process is a continuous semimartingale, 
allows us to price any variance swap. 
Nevertheless, out-of-money Call and Put option prices are needed for all possible strikes. 
Since this information is not available in practice, one can either adopt 
a model-free valuation formula or directly propose a parameterisation for the implied volatility surface.
The first approach involves a discretisation of the static replication formula,
which is how the VIX index is computed by the  Chicago Board Options Exchange. 
It is not easy, however, to extend this computation for large maturities (VIX is a $30$-day ahead index), 
where liquidity of options may play a major role. 
On the other hand, the second approach allows to calibrate a model 
using available data and additionally allows to interpolate / extrapolate 
available option data to all strikes and maturities.
In this work we follow the latter approach, 
with the eSSVI parameterisation~\cite{HM} for the implied volatility surface, 
which is a refinement of the SSVI parametrisation introduced in~\cite{Gatheral-Jacquier}:
\begin{equation}\label{eq:eSSVI}
\sigma_{\BS}^2(t,k)t = w(t, k)
 := \frac{\theta_t}{2} \left\{1+ \rho(\theta_t) \varphi(\theta_t) k + 
 \sqrt{ \left(\varphi(\theta_t) k+\rho(\theta_t)\right)^2+ 1-\rho(\theta_t)^2}\right\},
\end{equation}
where $\theta_t$ is the observed ATM variance curve,
and where the shape function~$\varphi(\cdot)$ takes the form
$\varphi(\theta) = \eta \theta^{-\lambda}(1+\theta)^{\lambda-1}$.
For the correlation parameter~$\rho(\cdot)$ we restrict it to the following functional form:
\begin{equation}\label{rho}
\rho(\theta) = (A-C)\E^{-B\theta} + C,
\qquad\text{for } (A, C) \in (-1,1)^2, B\geq 0,
\end{equation}
ensuring that $|\rho(\cdot)|\leq 1$.
We shall indistinctly refer to the total implied variance by~$\sigma_{\BS}^2(\cdot)t$ or~$w(\cdot)$, 
where~$\sigma_{\BS}(\cdot)$ represents the implied volatility. 
Gatheral and Jacquier~\cite{Gatheral-Jacquier} found (sufficient and almost necessary) conditions 
on the parameters $\rho$, $\varphi(\cdot)$, $\eta$ and~$\lambda$ preventing arbitrage. 
In the eSSVI formulation, the correlation has a term structure~\eqref{rho}, 
and care must be taken in order to preclude arbitrage. 
Concretely, following ~\cite{HM}, the restriction
\begin{equation}\label{eSSVICSCond}
|\theta\partial_{\theta}(\rho(\theta)) + \rho(\theta)\gamma|\leq \gamma,
\end{equation}
where
$\gamma := \partial_{\theta}(\theta\varphi(\theta)) / \varphi(\theta)$,
is a necessary condition to preclude calendar spread arbitrage.
To prevent butterfly arbitrage, exactly as in the SSVI parameterisation, it is sufficient~\cite[Theorem 4.2]{Gatheral-Jacquier} to check that the inequality 
$\theta_t \varphi^2(\theta_t) (1+|\rho(\theta_{t})|)\leq 4$ holds for all maturity~$t$.
In this parameterisation, variance swaps can be computed in closed form, as proved by Martini~\cite{Martini},
based on earlier works by Gatheral~\cite{Gatheral} and Fukasawa~\cite{Fukasawa}:
\begin{proposition}\label{prop:swapsSSVI}
The fair strike (in total variance) of a variance swap in the eSSVI model reads
$$
\sigma_0(t)^2 t := -2\EE\log\left(\frac{S_t}{S_0}\right) = \frac{b_{t}^2+2a_{t}(c_{t}+\theta_t)}{2a_{t}^2},
$$
where
$\chi_t := \frac{1}{4}[1-\rho(\theta_t)^2] \theta_t \varphi(\theta_t)$, and 
$$
a_{t} = 1+\frac{\theta_t \varphi(\theta_t)}{2}
\left(\rho(\theta_t)- \frac{\chi_t}{2}\right),
\qquad\qquad
b_{t} = \theta_t \varphi(\theta_t)\left[\chi_t-\rho(\theta_t)\right],
\qquad\qquad
c_{t} =\theta_t \varphi(\theta_t)\chi_t.
$$
\end{proposition}
Recalling the relation between variance swaps and the forward variance curve, we have
$$\xi_0(t)=\frac{\D}{\D t}\left(t \sigma^2_0(t)\right)=\sigma^2_0(t)+t\frac{d}{\D t}\sigma^2_0(t).$$
\begin{remark}
In order to interpolate/extrapolate the eSSVI, it is necessary to also interpolate/extrapolate~$\theta_t$ 
for all~$t$. 
However, $\theta_t$ is only observed on a discrete set of maturity dates,
and consequently, cubic splines are used to interpolate/extrapolate all other maturities.
\end{remark}


\subsubsection{Calibration algorithm and numerical results}
We first calibrate the eSSVI parameterisation~\eqref{eq:eSSVI} on the SPX implied volatility surface
on the~$4$th of December $2015$.
Figure~\ref{SSVI_calibr} shows the fit for the shortest, medium and longest maturities available in the data set. 
For short maturities the eSSVI is not able to fully capture the volatility smile, 
however as maturity increases the fit improves remarkably.
\begin{figure}[h]
\centering
\includegraphics[height=10cm, width=10cm]{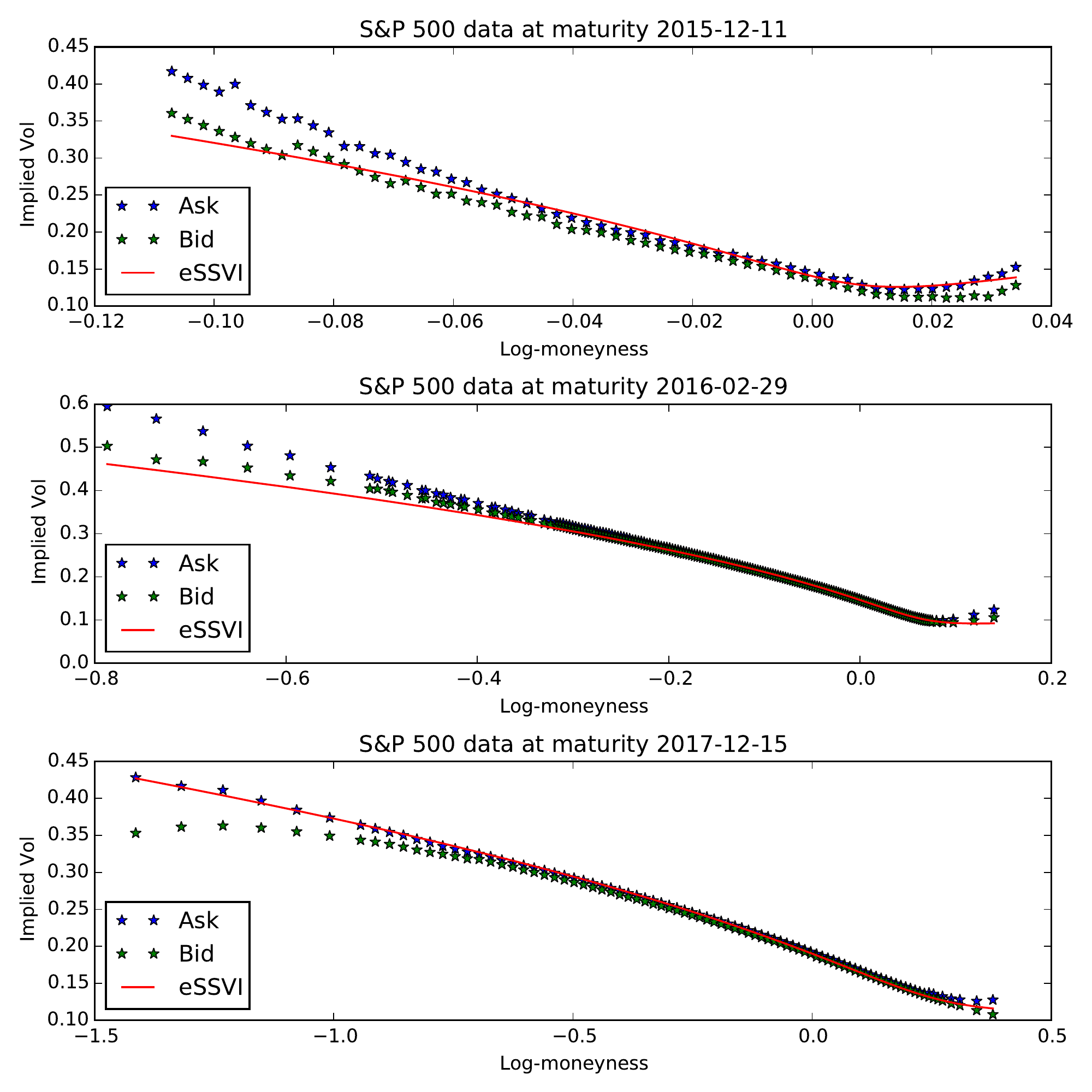}
\caption{eSSVI calibration results on 4/12/2015 using traded SPX options.}
\label{SSVI_calibr}
\end{figure}
For VIX Futures, the calibration algorithm reads as follows:
\begin{algorithm}[VIX Futures calibration algorithm in the rough Bergomi model]\label{algo:VIXFutures}\ 
\begin{enumerate}[(i)]
\item Calibrate eSSVI to available SPX option data;
\item compute the variance swap term structure $(\sigma_0(t)^2)_{t\geq 0}$ using Proposition~\ref{prop:swapsSSVI};
\item extract the initial forward variance curve, $\xi_0(\cdot)$ via
$\xi_0(t)\approx \sigma^2_0(t) + \frac{\sigma^2_0(t+\eps)-\sigma^2_0(t-\eps)}{2\eps}t$
(with $\eps = 1E-8$);
\item minimise (over $\nu,H$) the objective function in~\eqref{eq:ObjectFunc}.
\end{enumerate}
\end{algorithm}

Figures~\ref{VIX_calibr}-\ref{VIX_calibr3} suggest that the model fits very well the observed VIX Futures term structure for different dates.
Moreover, we notice that both the model and the observed data are qualitatively equal in terms of convexity/concavity. 
In the rough Bergomi model this information is obtained from option prices through~$\xi_0$, 
which suggests a correspondence in the market between VIX futures and SPX options. 
However, we also observe that in all three cases the error is greater for short maturities,
mimicking the calibration limits of eSSVI for short maturities, as detailed in Section~\ref{sec:InitialFwdCurve}.
\begin{figure}[h]
\centering
\includegraphics[scale=0.2]{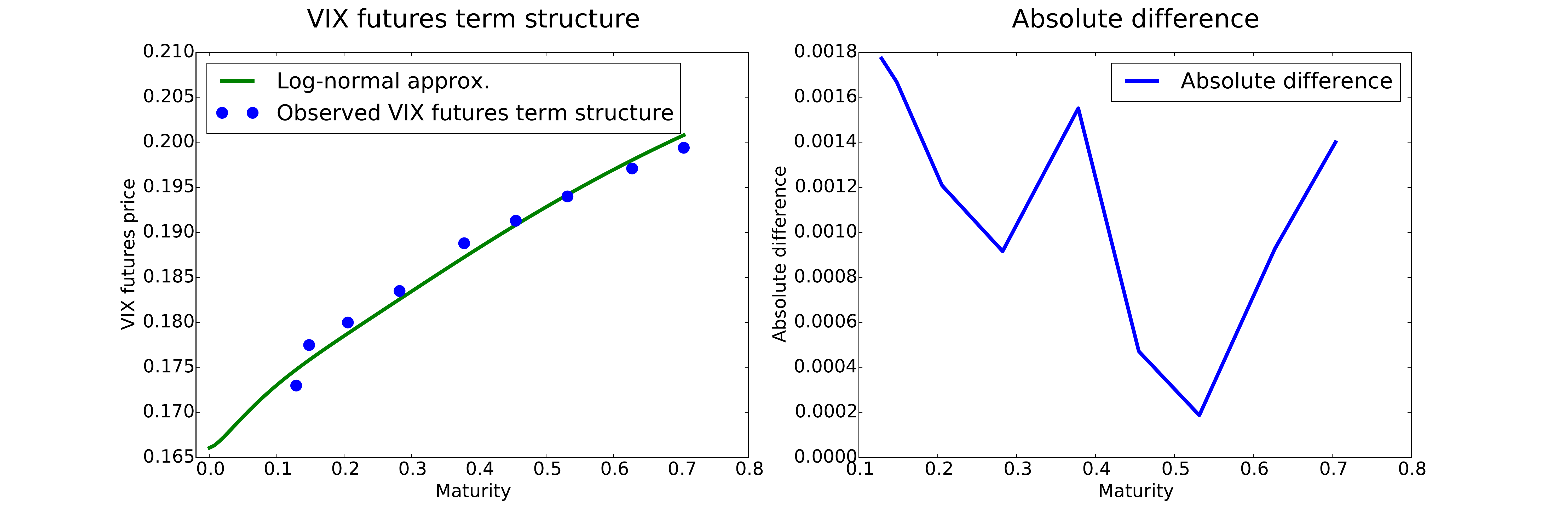}
\caption{VIX Futures calibration on 4/12/2015.
Optimal parameters: $(H, \nu) =  (0.09237, 1.004)$.}
\label{VIX_calibr}
\end{figure}

\begin{figure}[h]
\centering
\includegraphics[scale=0.2]{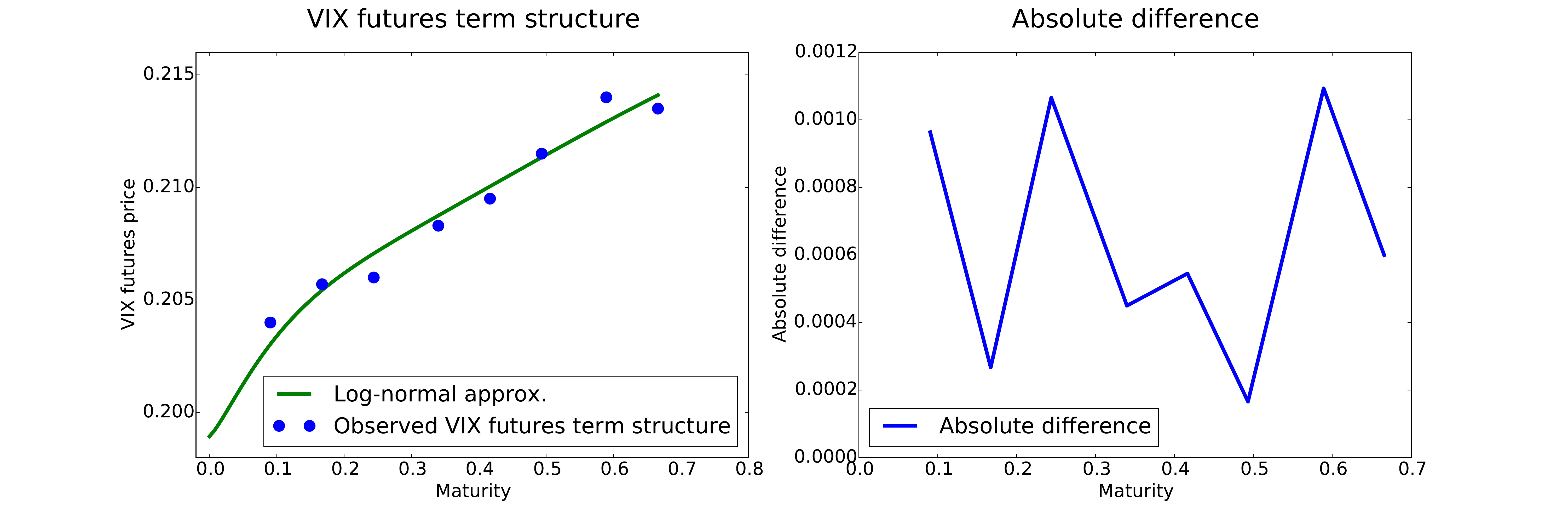}
\caption{VIX Futures calibration on 22/2/2016.
Optimal parameters: $(H, \nu) =  (0.10093, 1.00282)$.}
\label{VIX_calibr1}
\end{figure}

\begin{figure}[h]
\centering
\includegraphics[scale=0.2]{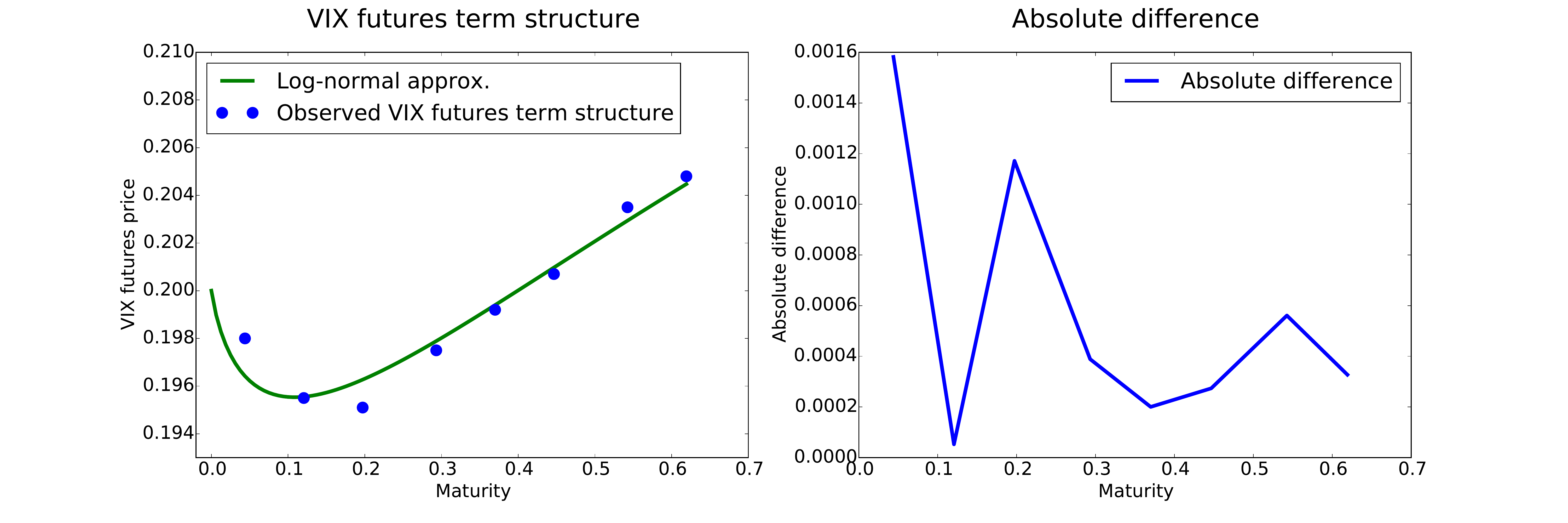}
\caption{VIX Futures calibration on 4/1/2016.
Optimal parameters: $(H, \nu) =  (0.0509, 1.2937)$.}
\label{VIX_calibr3}
\end{figure}

\begin{remark}
The reader should recall the importance of the initial forward variance curve $\left(\xi_0(t)\right)_{t\geq 0}$ 
in the VIX Futures process, since $(\Vk_t)_{t\geq 0}$ depends on the whole path of~$\xi_0$ up to time~$t$. Therefore, even if~$\xi_0$ is only misspecified for short maturities (as is the case of the eSSVI), 
this affects the whole term structure of the VIX process 
(for details we refer the reader to Section~\ref{sec:VIXFutures}). 
Therefore, an improved~$\xi_0$ estimation would not only increase the accuracy of the model 
for short maturities, but also for the whole term structure.
\end{remark}

\begin{remark}
Our calibration involves two different data sets, SPX options and VIX Futures:
the Vanilla quotes are extracted from the CBOE delayed option quotes page\footnote{CBOE delayed option quotes: http://www.cboe.com/delayedquote/quotetable.aspx} and the VIX Futures from the CBOE VIX futures historical data page\footnote{CBOE VIX futures historical data: http://cfe.cboe.com/data/historicaldata.aspx}. 
We perform an aggregation of the different Future quotes, since they are quoted on a Future per Future basis,
and check the consistency between the two data sets by comparing the left end extrapolation of 
the Futures curve with the theoretical VIX computed from option prices.
\end{remark}

\section{From VIX Futures to SPX options}\label{sec:JointCalib}
In this final chapter we assess whether the Hurst parameter~$H$ 
obtained through the VIX Futures calibration algorithm is consistent with SPX options. 
For this purpose, we calibrate the rough Bergomi model to SPX option data 
by fixing the parameter~$H$ and letting the algorithm calibrate~$\nu$ and~$\rho$. 
One of the main reasons to fix~$H$ is that the hybrid scheme introduced 
in Section~\ref{Hybrid simulation scheme} remarkably reduces its complexity to~$\Oo(n)$, 
since the $\Oo(n\log n)$ complexity of the Volterra is computed only once and reused afterwards. 
Therefore, by fixing~$H$ the pricing scheme is much faster when several valuations are performed, 
as is the case of a calibration algorithm.

\subsection{Pricing in the rough Bergomi model}\label{Pricing under rBergomi}
We present a pricing scheme, where the Volterra process~$\Vv$ is simulated using a hybrid scheme, 
while a standard Euler scheme generates the paths of the stock process:

\begin{algorithm}[Simulation of the rough Bergomi model]\label{alg:rBergomiSimul}
Consider the grid $\Tt:=\{t_i\}_{i=0,\ldots,n_T}$, and fix $\kappa\geq 1$. 
\begin{enumerate}[(i)]
\item Simulate the Volterra process~$\Vv$ on the grid~$\Tt$ using the hybrid scheme;
\item simulate the variance process as 
$V_t = \xi_0(t)\Ee(2\nu C_{H}\Vv_t)$, for $t\in\Tt$ and where $([\Vv]_t)_{t\geq 0}$ 
is given in~\eqref{eq: quadratic_variation};
\item extract the path of the Brownian motion~$Z$ driving~$\Vv$:
\begin{equation*}
\begin{array}{rll}
Z_{t_{i}}& =Z_{t_{i-1}}+n^{\Hm}\left(\Vv(t_{i}) - \Vv(t_{i-1})\right),& \text{for }i=1,\ldots,\kappa,\\
Z_{t_{i}}&=Z_{t_{i-1}}+\overline{Z}_{i-1} ,   & \text{for }i>\kappa;
  \end{array}
\end{equation*}
compute $\{Z^{\perp}\}_{i=0}^{n_T-1}$ where $Z^{\perp}_i \equalDistrib \Nn(0,1/n_T)$ is an independent standard Gaussian sample;
\item correlate the two Brownian motions via
$W_{t_{i}}-W_{t_{i-1}}
 = \rho \overline{Z}_{i-1} + \sqrt{1-\rho^2}\overline{Z}_{i-1}^{\perp};$
\item simulate $S_{t_{i}}=\exp(X_{t_i})$ using a forward Euler scheme:
$$
X_{t_{i+1}} = X_{t_{i}}-\frac{1}{2}V_{t_{i}}(t_{i+1}-t_{i})+\sqrt{V_{t_{i}}} \left(W_{t_{i+1}}-W_{t_{i}}\right), \qquad \text{for }  i=0,\ldots,n_T-1;
$$
\item compute the expectation by averaging the payoff over all terminal values of each path.
\end{enumerate}
\end{algorithm}


\subsection{Calibration of SPX options via VIX Futures}
We first follow the calibration algorithm~\ref{algo:VIXFutures} to obtain~$H$ and~$\xi_0$,
and we then aim at minimising, over~$(\nu,\rho)$, the objective function
\begin{equation}\label{eq:JointObjecFunc}
\Ll^{\Cr}(\nu,\rho) := \sum_{j=1}^{L}\sum_{i=1}^{N} (\Cr_{T_{i,j}}-\Cr^\mathrm{obs}_{i,j})^2,
\end{equation}
where $\Cr_{T_{i,j}}$ is the Call price given by the rough Bergomi model, 
computed using the scheme introduced in Section~\ref{Pricing under rBergomi}, 
with maturity~$T_i$ and strike~$K^{(j)}$. 
On the other hand, $(\Cr^\mathrm{obs}_{i,j})_{i,j}$ 
is the set observed Call prices in the time grid $T_1<\ldots<T_N$ and strike grid $K^{(1)}<\ldots<K^{(L)}$.
In order to optimise the calibration algorithm, we first compute the Volterra process~$\Vv$ ,
which will then be used in a forward Euler simulation at each calibration step:
\begin{algorithm}[Calibration algorithm for SPX options via VIX Futures]\ 
\begin{enumerate}[(i)]
\item Calibrate $H$ and $\xi_0$ using the VIX Futures;
\item compute $M$ paths of the Volterra process, 
$\{\Vv^{(u)}\}_{u=1}^{M}$ and extract the Brownian motions~$\{Z^{(u)}\}_{u=1}^M$ driving each process. Also, compute independent Brownian motions $\{Z^{\perp(u)}\}_{u=1}^M$;
\item evaluate the Call prices in each calibration step:
\begin{equation*}
\begin{array}{rll}
V_t^{(u)} & = \xi_0(t)\Ee\left(2\nu C_{H}\Vv_t^{(u)}\right), & u=1,\ldots,M,\\
W^{(u)} & = \rho Z^{(u)}+\sqrt{1-\rho^2}Z^{\perp(u)}, &  u=1,\ldots,M,\\
S^{(u)}_{t+\Delta} & = S_t^{(u)}+S_t^{(u)}\sqrt{v^{(u)}_t} \left(W^{(u)}_{t+\Delta}-W^{(u)}_{t}\right), & u=1,\ldots,M;
\end{array}
\end{equation*}
\item compute the Call price for each available maturity $\{T_1,\ldots,T_N\}$ 
and set of strikes $\{K^{(1)},\ldots,K^{(L)}\}$:
$$
\Cr_{T_{i},j}=\frac{1}{M}\sum_{u=1}^{M}(S^{(u)}_{T_i}-K^{(j)})_+,
\quad\text{for }i=1,\ldots,N \text{ and }j=1,\ldots,L;
$$
\item minimise over $(\nu,\rho)$ the objective function~$\Ll^{\Cr}(\nu,\rho)$ in~\eqref{eq:JointObjecFunc}.
\end{enumerate}
\end{algorithm}

\begin{remark}
Item~(v) in the algorithm above may change the optimal values for~$\nu$, which was initially calibrated in~(i) 
from the VIX Futures.
Backtesting however shows that the calibration in~(i) is not really affected by this.
\end{remark}

\subsection{Results}
We calibrate the model on December $4$, $2015$,
fixing $H=0.09237$ obtained previously through VIX, 
and plot the fit in Figure~\ref{options}. 
The model is not fully consistent for short maturities, which may follow from the inability 
of~$\xi_0$ to fully capture the smiles for these maturities, 
but the fit greatly improves with maturity. 
Interestingly, we observe a~$20\%$ difference between the the parameter~$\nu$ 
obtained through VIX calibration and the one obtained through SPX. 
This suggests that the volatility of volatility in the SPX market is~$20\%$ higher 
when compared to VIX, revealing potential data inconsistencies (arbitrage?).
Nevertheless, we emphasise the importance of an accurate~$\xi_0$ curve 
to improve the fit to SPX and to provide an efficient joint calibration.
\begin{figure}[h]
\includegraphics[scale=0.4]{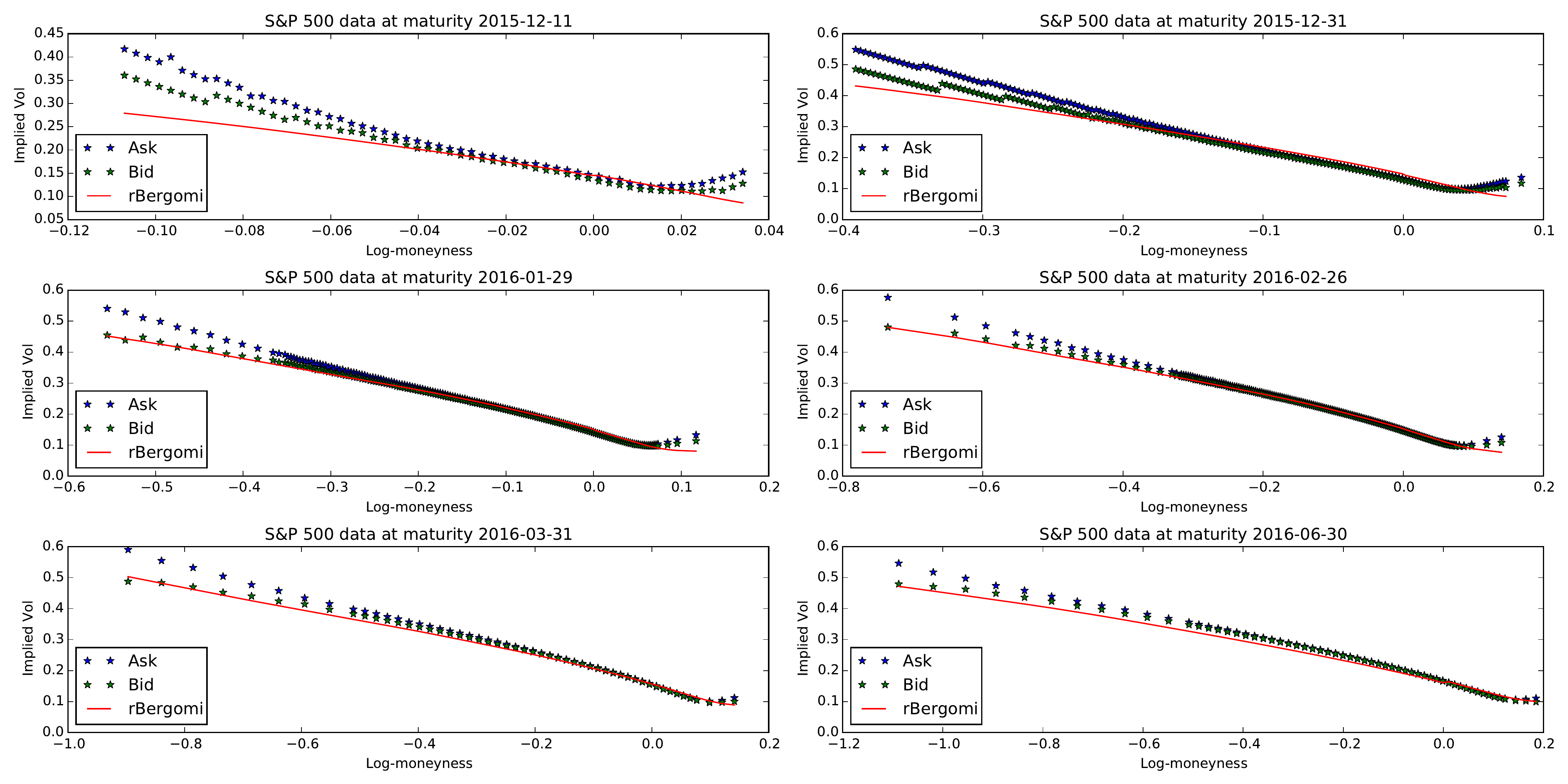}
\caption{Calibration of SPX smiles on 4/12/2015.
Calibrated parameters: $(\nu, \rho) = (1.19, -0.999)$.}
\label{options}
\end{figure}

\section*{Conclusion}
Following the path set by Bayer, Friz and Gatheral~\cite{BFG15},
we developed here a relatively fast algorithm to calibrate VIX Futures and the VIX smile,
consistently with the SPX smile, in the rough Bergomi model.
The clear strength of this model is that only a few parameters are needed, 
making the (re)calibration robust and stable.
From a trader's point of view, we highlight some potential market discrepancy between the VIX and the SPX,
and leave a refined analysis thereof for future research.

\appendix
\section{The hybrid scheme}\label{sec:HybridScheme}
We briefly recall the hybrid scheme developed in~\cite{BLP15}.
Following the notation in Definition~\ref{def:semistationary}, 
we consider a (truncated) Brownian semistationary process $\Bb(\alpha, W)$, 
and introduce the truncation parameter $\kappa\in\NN$.
On an equidistant grid $\Tt:=\{t_i=i/n\}_{i=0,\ldots,n_T}$, with $n_T:=\lfloor nT\rfloor$,
for $n\geq 2$,
under Definition~\ref{def:semistationary}, the hybrid scheme for the BSS process~$\Bb$ is approximated by
$\Bb_n(t_{i}) = \widetilde{\Bb}_n(i) + \widehat{\Bb}_n(i)$ with
$$
\widetilde{\Bb}_n(i)
 = \sum_{k=1}^{i\wedge \kappa} L_g\left(\frac{k}{n}\right)\sigma\left(\frac{i-k}{n}\right)\overline{W}_{i-k,k}
 \qquad\text{and}\qquad
\widehat{\Bb}_n(i)
 = \sum_{k=\kappa+1}^{i}g\left(\frac{b^*_k}{n}\right)\sigma\left(\frac{i-k}{n}\right)\overline{W}_{i-k},
$$
with $L_g$ introduced in Definition~\ref{def:semistationary}, and where
\begin{equation}\label{eq:bStar}
\overline{W}_{i} := \int_{t_{i}}^{t_{i+1}}\D W_s,
\quad
\overline{W}_{i,k} := \int_{t_{i}}^{t_{i+1}}\left(t_{i+k}-s\right)^{\alpha}\D W_s,
\quad
b_k^*=\left(\frac{k^{\alpha+1}-(k-1)^{\alpha+1}}{\alpha +1}\right)^{1/\alpha},
\text{ for }k\geq\kappa+1.
\end{equation}
For any $i, k$, the random variables $\overline{W}_i$ and $\overline{W}_{i,k}$ are centred Gaussian 
with the following covariance structure:
\begin{align*}
\EE\left(\overline{W}_{i,k}\overline{W}_{i}\right) & = \frac{k^{\alpha+1}-(k-1)^{\alpha+1}}{n^{\alpha+1}\alpha+1},
\qquad \text{and}\qquad 
\EE\left(\overline{W}_{i,k}\overline{W}_{j}\right)=0, \quad\text{for }k\neq j,\\
\EE\left(\overline{W}_{i,k}\overline{W}_{i,j}\right) & = \int_{0}^{1/n}\left(\frac{k}{n}-u\right)^{\alpha}\left(\frac{j}{n}-u\right)^\alpha \D u,\text{ for }k\neq j,
\quad \VV\left(\overline{W}_{i,k}\right) = \frac{k^{2\alpha+1}-(k-1)^{2\alpha+1}}{n^{2\alpha+1}2\alpha+1},
\quad 
\VV\left(\overline{W}_{i}\right) = \frac{1}{n}.
\end{align*}


\end{document}